\documentclass{article}

\usepackage{amsmath, amssymb, amsthm}
\usepackage{algorithm}
\usepackage{algorithmic}
\usepackage{xcolor}
\usepackage{bbm}
\usepackage{bm}
\usepackage{enumitem}
\usepackage{natbib}
\usepackage{fullpage}
\usepackage{graphicx}

\newtheorem{mydef}{Definition}
\newtheorem{theorem}{Theorem}
\newtheorem{lemma}{Lemma}

\newtheorem{corollary}{Corollary}

\newenvironment{nscenter}
 {\parskip=0pt\par\nopagebreak\centering}
 {\par\noindent\ignorespacesafterend}

\definecolor{DarkGreen}{rgb}{0.1,0.5,0.1}
\newcommand{\todo}[1]{\textcolor{DarkGreen}{[To do: #1]}}
\newcommand{\rc}[1]{\textcolor{blue}{[Rachel: #1]}}

\newcommand{\ignore}[1]{}

\renewcommand{\vec}[1]{{\textbf{#1}}}

\newcommand{\x}{\vec{x}}
\newcommand{\s}{\vec{s}}
\newcommand{\one}{\mathbbm{1}}
\newcommand{\dist}{\mathcal{D}}
\newcommand{\reals}{\mathbb{R}}
\renewcommand{\Pr}{\mathrm{Pr}}
\newcommand{\T}{\mathcal{T}}

\newcommand{\p}{\textbf{p}}
\newcommand{\m}{\textbf{m}}
\newcommand{\N}{\mathcal{N}}
\newcommand{\vPi}{\boldsymbol{\Pi}}
\newcommand{\bbN}{\mathbb{N}}

\newcommand{\msize}{20pt}

\newcommand{\eps}{\epsilon}
\newcommand{\om}{\omega}
\newcommand{\gam}{\gamma}

\newcommand{\E}{\mathbb{E}}
\newcommand{\R}{\mathbb{R}}

\newcommand{\cS}{\mathcal{S}}
\newcommand{\dom}{\mathcal{I}}
\newcommand{\cM}{\mathcal{M}}

\newcommand{\vp}{\textbf{p}}
\newcommand{\vm}{\textbf{m}}

\newcommand{\vx}{\textbf{x}}

\newcommand{\Prof}{\mathop{\mathrm{Profit}}}

\title{The Possibilities and Limitations of Private Prediction Markets}

\author{Rachel Cummings\thanks{Computing and Mathematical Sciences, California Institute of Technology. Email: \texttt{rachelc@caltech.edu}  Supported in part by a Simons Award for Graduate Students in Theoretical Computer Science.  Much of this research was done while R. Cummings was at Microsoft Research.}
\and
David M.\ Pennock\thanks{Microsoft Research. Email: \texttt{dpennock@microsoft.com}}
\and
Jennifer Wortman Vaughan\thanks{Microsoft Research. Email: \texttt{jenn@microsoft.com}}
}

\begin{document}

\maketitle

\begin{abstract}
  We consider the design of \emph{private prediction markets},
  financial markets designed to elicit predictions about uncertain
  events without revealing too much information about market
  participants' actions or beliefs.  Our goal is to design market
  mechanisms in which participants' trades or wagers influence the
  market's behavior in a way that leads to accurate predictions, yet
  no single participant has too much influence over what others are
  able to observe.  We study the possibilities and limitations of such
  mechanisms using tools from differential privacy.  We begin by
  designing a private one-shot wagering mechanism in which bettors
  specify a belief about the likelihood of a future event and a
  corresponding monetary wager. Wagers are redistributed among bettors
  in a way that more highly rewards those with accurate predictions.
  We provide a class of wagering mechanisms that are guaranteed to
  satisfy truthfulness, budget balance in expectation, and other
  desirable properties while additionally guaranteeing
  $\epsilon$-joint differential privacy in the bettors' reported
  beliefs, and analyze the trade-off between the achievable level of
  privacy and the sensitivity of a bettor's payment to her own report.
  We then ask whether it is possible to obtain privacy in dynamic
  prediction markets, focusing our attention on the popular
  cost-function framework in which securities with payments linked to
  future events are bought and sold by an automated market maker.  We
  show that under general conditions, it is impossible for such a
  market maker to simultaneously achieve bounded worst-case loss and
  $\epsilon$-differential privacy without allowing the privacy
  guarantee to degrade extremely quickly as the number of trades
  grows, making such markets impractical in settings in which privacy
  is valued.  We conclude by suggesting several avenues for
  potentially circumventing this lower bound.
\end{abstract}

\section{Introduction}
\label{sec:intro}

Betting markets of various forms---including the stock
exchange~\citep{Gro:76}, futures markets~\citep{Roll:84}, sports betting
markets~\citep{Gan:98}, and markets at the
racetrack~\citep{Tha:88}---have been shown to successfully collect and
aggregate information.  Over the last few decades, \emph{prediction
  markets} designed specifically for the purpose of elicitation and
aggregation, have yielded useful predictions in domains as diverse as
politics~\citep{Ber:01}, disease surveillance~\citep{Polgreen:2007}, and
entertainment~\citep{PenScience:01}.

The desire to aggregate and act on the strategically valuable
information dispersed among employees has led many companies to
experiment with internal prediction markets.  An internal corporate
market could be used to predict the launch date of a new product or
the product's eventual success.  Among the first companies to
experiment with internal markets were Hewlett-Packard, which
implemented real-money markets, and Google, which ran markets using
its own internal currency that could be exchanged for raffle tickets
or prizes~\citep{CP02,CZ15}. More recently, Microsoft, Intel, Ford, GE,
Siemens, and others have engaged in similar
experiments~\citep{BP2009,C07,CZ15}.

Proponents of internal corporate markets often argue that the market
structure helps in part because, without it, ``business
practices... create incentives for individuals not to reveal their
information''~\citep{CP02}.  However, even with a formal market
structure in place, an employee might be hesitant to bet against the
success of their team for fear of insulting her coworkers or angering
management. If an employee has information that is unfavorable to
the company, she might choose not to report it, leading to predictions
that are overly optimistic for the company and ultimately contributing
to an ``optimism bias'' in the market similar to the bias in Google's
corporate markets discovered by \citet{CZ15}.

To address this issue, we consider the problem of designing
\emph{private} prediction markets.  A private market would allow
participants to engage in the market and contribute to the accuracy of
the market's predictions without fear of having their information or
beliefs revealed.  The goal is to provide participants with a form of
``plausible deniability.''  Although participants' trades or wagers
should together influence the market's behavior and predictions, no
single participant's actions should have too much influence over what
others can observe.  We formalize this idea using the popular notion
of \emph{differential privacy}~\citep{DMNS06,DR14}, which can be used
to guarantee that any participant's actions cannot be inferred from
observations.

We begin by designing a private analog of the \emph{weighted score
  wagering mechanisms} first introduced by \citet{LLW+08}.  A
\emph{wagering mechanism} allows bettors to each specify a belief
about the likelihood of a future event and a corresponding monetary
wager.  These wagers are then collected by a centralized operator and
redistributed among bettors in such a way that more accurate bettors
receive higher rewards. \citet{LLW+08} showed that the class of
weighted score wagering mechanisms, which are built on the machinery
of proper scoring rules~\citep{Gneiting:07}, is the unique set of
wagering mechanisms to satisfy a set of desired properties such as
budget balance, truthfulness, and anonymity.  We design a class of
wagering mechanisms with randomized payments that maintain the nice
properties of weighted score wagering mechanisms in expectation while
additionally guaranteeing $\epsilon$-joint differential privacy in the
bettors' reported beliefs.  We discuss the trade-offs that exist
between the privacy of the mechanism (captured by the parameter
$\epsilon$) and the sensitivity of a bettor's payment to her own
report, and show how to set the parameters of our mechanisms to
achieve a reasonable level of the plausible deniability desired in
practice.

We next address the problem of running private dynamic prediction
markets.  We consider the setting in which traders buy and sell
securities with values linked to future events.  For example, a market
might offer a security worth $\$1$ if Microsoft Bing's market share
increases in 2016 and $\$0$ otherwise.  A risk neutral trader who
believes that the probability of Bing's market share increasing is $p$
would profit from buying this security at any price less than $\$p$ or
(short) selling it at any price greater than $\$p$.  The market price
of the security is thought to reflect traders' collective beliefs
about the likelihood of this event.  We focus on \emph{cost-function
  prediction markets}~\citep{CP07,ACV13} such as Hanson's popular
logarithmic market scoring rule~\citep{H03}.  In a cost-function
market, all trades are placed through an \emph{automated market
  maker}, a centralized algorithmic agent that is always willing to
buy or sell securities at some current market price that depends on
the history of trade via a potential function called the cost
function. We ask whether it is possible for a market maker to price
trades according to a noisy cost function in a way that maintains
traders' privacy without allowing traders to make unbounded profit off
of the noise.  Unfortunately, we show that under general assumptions,
it is impossible for a market maker to achieve bounded loss and
$\epsilon$-differential privacy without allowing the privacy guarantee
to degrade very quickly as the number of trades grows.  In particular,
the quantity $e^\epsilon$ must grown faster than linearly in the
number of trades, making such markets impractical in settings in which
privacy is valued.  We suggest several avenues for future research
aimed at circumventing this lower bound.

There is very little prior work on the design of private prediction
markets, and to the best of our knowledge, we are the first to
consider privacy for one-shot wagering mechanisms.  Most closely
related to our work is the recent paper of \citet{WFA15} who consider
a setting in which each of a set of self-interested agents holds a
private data point consisting of an observation $x$ and corresponding
label $y$.  A firm would like to purchase the agents' data in order to
learn a function to accurately predict the labels of new observations.
Building on the mathematical foundations of cost-function market
makers, \citeauthor{WFA15} propose a mechanism that provides
incentives for the agents to reveal their data to the firm in such a
way that the firm is able to solve its prediction task while
maintaining the agents' privacy.
The authors mention that similar ideas can be applied to produce
privacy-preserving prediction markets, but their construction requires
knowing the number of trades that will occur in advance to
appropriately set parameters.  The most straightforward way of
applying their techniques to prediction markets results in a market
maker falling in the class covered by our impossibility result,
suggesting that such techniques cannot be used to derive a
privacy-preserving market with bounded loss when the number trades is
unknown.

\section{Tools from Differential Privacy}
\label{sec:prelims}
\label{sec:privacyprelims} 

We formalize privacy using the now-standard notion of
\emph{differential privacy}, which was introduced by
\citet{DMNS06}. The most basic version of differential privacy is used
to measure the privacy of a randomized algorithm's output when given
as input a database $D$ with $n$ entries from some input domain
$\dom$.  Differential privacy is often studied in settings in which
the $n$ entries are provided by $n$ agents, each of whom would like to
keep their entry private. 
Two databases $D$ and $D'$ are said to be \emph{neighboring} if they
differ only in a single entry.  Differential privacy requires that the
distribution of the algorithm's output given $D$ is ``close to'' the
distribution of its output given any neighboring database $D'$.  For
these definitions, it is enough to view an ``algorithm'' as a
randomized function mapping inputs to outputs; we are not concerned
with the precise way in which the outputs are computed.  In the
following definitions, we restrict to real-valued outputs for
consistency with the algorithms used in this paper.

\begin{mydef}[Differential Privacy~\citep{DMNS06}]
\label{def:dp}
For any $\epsilon, \delta \geq 0$, an algorithm $\cM:\dom^n\rightarrow \reals$ is
{\em$(\epsilon,\delta)$-differentially private} if for every pair of
neighboring databases $D, D' \in \dom^n$ and every subset $\mathcal{S}
\subseteq \reals$,
\[ 
\Pr[\cM(D) \in \mathcal{S}] \leq e^{\epsilon} \Pr[\cM(D') \in
\mathcal{S}] + \delta. 
\]
If $\delta = 0$, we say that $\cM$ is {\em$\epsilon$-differentially private}.
\end{mydef}

As $\epsilon$ approaches 0, it becomes increasingly more difficult to
distinguish neighboring databases, leading to a higher level of
privacy.  As $\epsilon$ grows large, the privacy guarantee grows
increasingly weak.  There is generally no consensus about what
constitutes a ``good'' value of $\epsilon$, and the strength of the
guarantee needed may depend on the application.  We discuss this point
more later in the context of our results.

We sometimes abuse terminology and say that a random variable (such as
a bettor's profit) is differentially private.  This should be taken to
mean that the algorithm used to compute the value of the random
variable is differentially private.

\ignore{
\begin{mydef}[Differential Privacy~\cite{DMNS06}]\label{def.dp}
  An algorithm $\cM:\dom^n\rightarrow \reals$ is
  {\em$(\epsilon,\delta)$-differentially private} if for every pair of
  neighboring databases $D, D' \in \dom^n$ and for every subset of
  possible outputs $\mathcal{S} \subseteq \reals$,
\[ \Pr[\cM(D) \in \mathcal{S}] \leq \exp(\epsilon)\Pr[\cM(D') \in \mathcal{S}] + \delta. \]
If $\delta = 0$, we say that $\cM$ is {\em $\epsilon$-differentially private}.
\end{mydef}
} 


In the context of mechanism design, differential privacy is often too
strong of a notion.  Suppose, for example, that the algorithm $\cM$
outputs a vector of prices that each of $n$ agents will pay based on
their joint input.  While we may want the price that agent $i$ pays to
be differentially private in the input of the \emph{other} agents, it
is natural to allow it to be more sensitive to changes in $i$'s own
input.  To capture this idea, \citet{KPRU14} defined the notion of
\emph{joint differential privacy}.  Call two neighboring databases $D$
and $D'$ \emph{$i$-neighbors} if they differ only in the $i$th entry.
Suppose that $\cM$ now outputs one element for each agent $i \in
\{1,\ldots,n\}$ and let $\cM(D)_{-i}$ denote the vector of outputs to
all agents excluding agent $i$.  Then joint differential privacy is
defined as follows.

\begin{mydef}[Joint Differential Privacy~\citep{KPRU14}]
\label{def:jdp}
For any $\epsilon, \delta \geq 0$, an algorithm $\cM:\dom^n \rightarrow \reals^n$ is
{\em $(\epsilon,\delta)$-joint differentially private} if for every $i \in
\{1,\ldots,n\}$, for every pair of $i$-neighbors $D, D' \in \dom^n$,
and for every subset $\mathcal{S} \subseteq \reals^{n-1}$,
\[ 
\Pr[\cM(D)_{-i} \in \mathcal{S}] \leq e^{\epsilon} \Pr[\cM(D')_{-i}
\in \mathcal{S}] + \delta. 
\]
If $\delta = 0$, we say that $\cM$ is {\em$\epsilon$-joint differentially private}.
\end{mydef}

\ignore{
\begin{mydef}[Joint Differential Privacy~\cite{KPRU14}]\label{def.jdp}
 An algorithm $\cM:\dom^n \rightarrow \reals^n$ is
  {\em $(\epsilon,\delta)$-joint differentially private} if for every
  $i$, for every pair of $i$-neighbors $D, D' \in \dom^n$, and for
  every subset of outputs $\mathcal{S} \subseteq \reals^{n-1}$,
\[ \Pr[\cM(D)_{-i} \in \mathcal{S}] \leq \exp(\epsilon)\Pr[\cM(D')_{-i} \in \mathcal{S}] + \delta. \]
If $\delta = 0$, we say that $\cM$ is {\em $\epsilon$-jointly differentially private}.
\end{mydef}
}

Joint differential privacy is still a strong requirement.  It protects
the privacy of any agent $i$ from arbitrary coalitions; even if all
other agents shared their private output, they would still not be able
to learn too much about the input of agent $i$.

One useful tool for proving joint differential privacy is the
\emph{billboard lemma} \citep{HHR+14}.  The idea behind the billboard
lemma is quite intuitive and simple.  Imagine that we display some
message publicly so that it is viewable by all $n$ agents, as if
posted on a billboard, and suppose that the algorithm to compute this
message is $\epsilon$-differentially private.  If each agent $i$'s
output $\cM(D)_i$ is computable from this public message along with
$i$'s own private input, then $\cM$ is $\epsilon$-joint differentially
private.  A more formal statement and proof are given in
\citet{HHR+14}.


\ignore{
\begin{lemma}[Billboard Lemma \cite{HHR+14}]\label{lem.billboard}
Suppose $\cM\colon \dom^n \rightarrow \reals$ is $(\eps, \delta)$-differentially private. Consider any set of functions $F_i\colon \dom_i\times \reals \rightarrow \reals'$, where $\dom_i$ is the $i$-th entry of the input data. The composition $\{F_i\left(\prod_iD, \cM(D)\right)\}$ is $(\eps, \delta)$-jointly differentially private, where $\prod_i$ is the projection to $i$'s data.
\end{lemma}
}

\label{sec:privstreams}

The definitions above assume the input database $D$ is fixed.
Differential privacy has also been considered for streaming
algorithms~\citep{CSS11,DNPR10}.  Let $\bbN = \{1, 2, 3, \ldots\}$.
Following~\citet{CSS11}, a stream $\sigma \in \dom^\bbN$ is a string
of countable length of elements in $\dom$, where $\sigma_t \in \dom$
denotes the element at position or \emph{time} $t$ and $\sigma_{1,
  \ldots, t} \in \dom^t$ is the length $t$
prefix of the stream $\sigma$.  Two streams $\sigma$ and $\sigma'$ are
said to be \emph{neighbors} if they differ at exactly one time $t$.

A streaming algorithm $\cM$ is said to be \emph{unbounded} if it
accepts streams of indefinite length, that is, if for any stream
$\sigma \in \dom^\bbN$, $\cM(\sigma) \in \reals^\bbN$.  In contrast, a
streaming algorithm is $T$-\emph{bounded} if it accepts only streams
of length at most $T$.  \citet{DNPR10} consider only $T$-bounded
streaming algorithms. Since we consider unbounded streaming
algorithms, we use a more appropriate definition of
differential privacy for streams adapted from \citet{CSS11}. For
unbounded streaming algorithms, it can be convenient to let the
privacy guarantee degrade as the input stream grows in length.
\citet{CSS11} implicitly allow this in some of their results; see, for
example, Corollary 4.5 in their paper.  For clarity and preciseness,
we explicitly capture this in our definition. Here and throughout the
paper we use $\reals_+$ to denote the nonnegative reals.

\ignore{
\begin{mydef}[Differential Privacy for Streams~\cite{CSS11}]
\label{def:streamdp}
For any $\epsilon > 0$, a streaming algorithm $\cM:\dom^\bbN
\rightarrow \reals^\bbN$ is {\em $\epsilon$-differentially private} if
for every pair of neighboring streams $\sigma, \sigma' \in \dom^\bbN$
and for every subset $\mathcal{S} \subseteq \reals^{\bbN}$,
\[ 
\Pr[\cM(\sigma) \in \mathcal{S}] \leq e^{\epsilon} \Pr[\cM(\sigma')
\in \mathcal{S}]. 
\]
\end{mydef}
} 

\ignore{
\begin{mydef}[Differential Privacy for Streams]
\label{def:streamdp}
For any non-decreasing functions
$\epsilon : \bbN \rightarrow \reals_{+}$ and
$\delta : \bbN \rightarrow \reals_{+}$, a streaming algorithm
$\cM:\dom^\bbN \rightarrow \reals^\bbN$ is {\em
  $(\epsilon(t), \delta(t))$-differentially private} if for every pair of
neighboring streams $\sigma, \sigma' \in \dom^\bbN$, for every
$t \in \bbN$, and for every subset $\mathcal{S} \subseteq \reals^{t}$,
\[ 
\Pr[\cM(\sigma_{1, \ldots, t}) \in \mathcal{S}] \leq e^{\epsilon(t)}
\Pr[\cM(\sigma'_{1, \ldots, t})
\in \mathcal{S}] + \delta(t). 
\]
If $\delta(t) = 0$ for all $t$, we say that $\cM$ is {\em $\epsilon(t)$-differentially private}.
\end{mydef}
} 

\begin{mydef}[Differential Privacy for Streams]
\label{def:streamdp}
For any non-decreasing function
$\epsilon : \bbN \rightarrow \reals_{+}$ and any
$\delta \geq 0$, a streaming algorithm
$\cM:\dom^\bbN \rightarrow \reals^\bbN$ is {\em
  $(\epsilon(t), \delta)$-differentially private} if for every pair of
neighboring streams $\sigma, \sigma' \in \dom^\bbN$, for every
$t \in \bbN$, and for every subset $\mathcal{S} \subseteq \reals^{t}$,
\[ 
\Pr[\cM(\sigma_{1,\ldots,t}) \in \mathcal{S}] \leq e^{\epsilon(t)} \Pr[\cM(\sigma'_{1,\ldots,t})
\in \mathcal{S}] + \delta. 
\]
If $\delta = 0$, we say that $\cM$ is {\em $\epsilon(t)$-differentially private}.
\end{mydef} 

Note that we allow $\epsilon$ to grow with $t$, but require that
$\delta$ stay constant.  In principle, one could also allow $\delta$
to depend on the length of the stream.  However, allowing $\delta$ to
increase would likely be unacceptable in scenarios in which privacy is
considered important.  In fact, it is more typical to require
\emph{smaller} values of $\delta$ for larger databases since for a
database of size $n$, an algorithm could be considered
$(\epsilon, \delta)$-private for $\delta$ on the order of $1/n$ even
if it fully reveals a small number of randomly chosen database
entries~\citep{DR14}.  Since we use this definition only when showing
an impossibility result, allowing $\delta$ to decrease in $t$ would
not strengthen our result.

\ignore{
We discuss how the particular streaming algorithms of \citet{CSS11}
and \citet{DNPR10} could be applied in the context of dynamic
prediction markets and the relationship to our lower bounds in
Section~\ref{sec:costfuncs}.
}

\ignore{
The mechanisms of \citet{CSS11} and \citet{DNPR10} show how to
maintain a noisy count of the number of ones in a stream of bits.
Their algorithms satisfy differential privacy along with guarantees on
the accuracy of the counts that degrade in quality with the length of
the input stream.  Both papers achieve this by computing the exact
count and adding noise that is correlated across time but independent
of the data.  This correlation in the noise allows for improved
accuracy of the noisy count over naively adding fresh noise at each
time $t$. 
}

\section{Private Wagering Mechanisms}
\label{sec:oneshot} 

We begin with the problem of designing a one-shot wagering mechanism
that incentivizes bettors to truthfully report their beliefs while
maintaining their privacy.  A wagering mechanism allows a set of
bettors to each specify a belief about a future event and a monetary
wager.  Wagers are collected by a centralized operator and
redistributed to bettors in such a way that bettors with more accurate
predictions are more highly rewarded. \citet{LLW+08} showed that the
class of \emph{weighted score wagering mechanisms} (WSWMs) is the
unique class of wagering mechanisms to satisfy a set of desired axioms
such as budget balance and truthfulness.  In this section, we show how
to design a randomized wagering mechanism that achieves
$\epsilon$-joint differential privacy while maintaining the nice
properties of WSWMs in expectation.

\subsection{Standard wagering mechanisms}

Wagering mechanisms, introduced by \citet{LLW+08}, are mechanisms
designed to allow a centralized operator to elicit the beliefs of a
set of bettors without taking on any risk.  In this paper we focus on
\emph{binary} wagering mechanisms, in which each bettor $i$ submits a
report $p_i \in [0,1]$ specifying how likely she believes it is that a
particular event will occur, along with a wager $m_i \geq 0$
specifying the maximum amount of money that she is willing to lose.
After all reports and wagers have been collected, all parties observe
the realized outcome $\omega \in \{0,1\}$ indicating whether or not
the event occurred.  Each bettor $i$ then receives a payment that is a
function of the outcome and the reports and wagers of all bettors.
This idea is formalized as follows.

\begin{mydef}[Wagering Mechanism~\citep{LLW+08}]
  A \emph{wagering mechanism} for a set of bettors $\N =
  \{1,\ldots,n\}$ is specified by a vector $\vPi$ of (possibly
  randomized) profit functions, $\Pi_i : [0,1]^n \times \reals^n_+
  \times \{0,1\} \to \R$, where $\Pi_i(\p, \vm, \om)$ denotes the
  total profit to bettor $i$ when the vectors of bettors' reported
  probabilities and wagers are $\p$ and $\vm$ and the
  realized outcome is $\om$.  It is required that $\Pi_i(\p, \vm, \om)
  \geq -m_i$ for all $\p$, $\vm$, and $\om$, which ensures that no
  bettor loses more than her wager.
\end{mydef}

There are two minor differences between the definition presented here
and that of \citet{LLW+08}.  First, for convenience, we use $\Pi_i$ to
denote the \emph{total} profit to bettor $i$ (i.e., her payment from
the mechanism minus her wager), unlike \citet{LLW+08}, who use $\Pi_i$
to denote the payment only.  While this difference is inconsequential,
we mention it to avoid confusion.  Second, all previous work on
wagering mechanisms has restricted attention to \emph{deterministic}
profit functions $\Pi_i$.  Since randomization is necessary to attain
privacy, we open up our study to \emph{randomized} profit functions.

\ignore{ 
  One property that we require this mechanism to have, is that
  bettors should not lose more than their wager, even if they make an
  incorrect prediction.  This is necessary for two reasons.  First,
  logistics: One can imagine this mechanism being run in two periods.
  In the first period, bettors ``buy in'' with their wager $m_i$, then
  the true outcome $\om$ is realized, and in the second period bettors
  each receive some payment based on the run of the mechanism.  If a
  bettor is expected to pay more than her wager, the designer (or
  market maker) has no means of extracting the additional payment from
  the bettor.  Second --- and more important to the incentives of
  bettors in the game --- each bettor $i$ receives utility $\Pi_i(\p,
  \vm, \om)$.  If a bettor risks losing more than her wager, then she
  may not have an incentive to even participate in the mechanism.
  Thus we will require that for all bettors $i$, and for all $\p$,
  $\vm$, $\om$, and randomness in the profit function, $\Pi_i(\p, \vm,
  \om) \geq -m_i$.  
}

\citet{LLW+08} defined a set of desirable properties or axioms that
deterministic wagering mechanisms should arguably satisfy. Here we
adapt those properties to potentially randomized wagering mechanisms,
making the smallest modifications possible to maintain the spirit of
the axioms.  Four of the properties (truthfulness, individual
rationality, normality, and monotonicity) were originally defined in
terms of expected profit with the expectation taken over some true or
believed distribution over the outcome $\omega$.  We allow the
expectation to be over the randomness in the profit function as well.
Sybilproofness was not initially defined in expectation; we now ask
that this property hold in expectation with respect to the randomness
in the profit function.  We define anonymity in terms of the
distribution over all bettors' profits, and ask that budget balance
hold for any realization of the randomness in $\vPi$.

\begin{itemize}[leftmargin=\msize]

\item[(a)] {\bf Budget balance:}
The operator makes no profit or loss, i.e., $\forall \p \in
[0,1]^n$, $\forall \vm \in \reals^n_+$, $\forall \om \in \{0,1\}$, and
for any realization of the randomness in $\vPi$,
$
\sum_{i =1}^n \Pi_i (\p, \vm, \om) = 0.
$
\item[(b)] {\bf Anonymity:} Profits do not depend on the identify of
  the bettors. That is, for any permutation of the bettors $\sigma$, $\forall \p \in
  [0,1]^n$, $\forall \vm \in \reals^n_+$, $\forall \om \in \{0,1\}$,
  the joint distribution over profit vectors $\{\Pi_i(\p, \vm,
  \omega)\}_{i \in \N}$ is the same as the joint distribution over
  profit vectors $\{\Pi_{\sigma(i)}\left( (p_{\sigma^{-1}(i)})_{i \in
      \N}, (m_{\sigma^{-1}(i)})_{i \in \N}, \om \right)\}_{i \in \N}$.
\item[(c)] {\bf Truthfulness:} Bettors uniquely maximize their
  expected profit by reporting the truth. That is,
  $\forall i \in \N$, $\forall \vp_{-i} \in [0,1]^{n-1}$,
  $\forall \vm \in \reals^n_+$, $\forall p^*, p_i \in [0,1]$ with $p_i \neq p^*$,
\[ 
\E_{\omega \sim p^*} \left[ \Pi_i( (p^*, \vp_{-i}), \vm, \om) \right] >
\E_{\omega \sim p^*} \left[ \Pi_i(
  (p_i, \vp_{-i}), \vm, \om) \right]. 
\]
\item[(d)] {\bf Individual rationality:} Bettors prefer participating
  to not participating. That is, $\forall i \in \N$, $\forall
  m_i >0$, for all $p^* \in [0,1]$, there exists some $p_i \in [0,1]$
  such that $\forall \p_{-i} \in [0,1]^{n-1}$, $\forall \vm_{-i} \in \reals^{n-1}_+$,
$
E_{\omega \sim p^*}\left[ \Pi_i( (p_i, \vp_{-i}), \vm, \om) \right] \geq
0. 
$
\item[(e)] {\bf Normality:\footnote{\citet{LL+15} and \citet{CDPV14}
      used an alternative definition of normality for wagering
      mechanisms that essentially requires that if, from some agent
      $i$'s perspective, the prediction of agent $j$ improves, then
      $i$'s expected profit decreases. This form of normality also
      holds for our mechanism.}}
%
\ignore{
 i.e., $\forall i, j \in \N$, $i \neq j$,
  $\forall p^* \in [0,1]$, $\forall \p \in [0,1]^n$, $\forall \vm \in
  \reals^n_+$, $\forall \om \in \{0,1\}$, let $\hat{\p}$ be defined
  with $\hat{p}_j = p^*$ and $\hat{p}_k = p_k$ for $k \neq j$. Then
\[ 
\E_{\omega \sim p^*} \left[ \Pi_i(\p, \vm, \om) \right] \geq
\E_{\omega \sim p^*} \left[ \Pi_i(\hat{\p}, \vm, \om) \right]. 
\]
} 
If any bettor $j$ changes her report, the change in the expected
profit to any other bettor $i$ with respect to a fixed belief $p^*$ is
the opposite sign of the change in expected payoff to $j$. That is, $\forall i, j \in \N$, $i \neq j$, $\forall \p, \p' \in [0,1]^n$
with $p'_k = p_k$ for all $k \neq j$,
$\forall p^* \in [0,1]$, $\forall \vm \in \reals^n_+$,
\[
\E[\Pi_j(\p, \vm, \om)] < \E[\Pi_j(\p', \vm, \om)]
\implies 
\E[\Pi_i(\p, \vm, \om)] \geq \E[\Pi_i(\p', \vm, \om)].
\]
All expectations are taken w.r.t. $\omega \sim p^*$ and the randomness
in the mechanism.
\item[(f)] {\bf Sybilproofness:} Profits remain unchanged as any
  subset of players with the same reports manipulate user accounts by
  merging accounts, creating fake identities, or transferring wagers. That is, $\forall \cS \subset \N$, $\forall \p$ with $p_i = p_j$ for all
  $i, j \in \cS$, $\forall \vm, \vm' \in \reals^n_+$ with $m_i = m_i'$
  for $i \notin \cS$ and $\sum_{i \in \cS} m_i = \sum_{i \in \cS}
  m_i'$, $\forall \omega \in \{0,1\}$, two conditions
  hold:
\vspace{-5pt}
\[ 
\E\left[\Pi_i(\p, \vm, \om)\right] = \E\left[\Pi_i(\p, \vm',
  \om)\right] \qquad \forall i \notin \cS,
\]
\[ 
\sum_{i \in \cS} \E\left[\Pi_i(\p, \vm, \om)\right] = \sum_{i \in \cS}
\E\left[\Pi_i(\p, \vm', \om)\right]. 
\]
\item[(g)] {\bf Monotonicity} The magnitude of a bettor's expected profit (or
  loss) increases as her wager increases. That is, $\forall i \in
  \N$, $\forall \p \in [0,1]^n$, $\forall \vm \in
  \reals^n_+$, $\forall M_i > m_i$, $\forall p^* \in [0,1]$, either
$
0 < \E_{\omega \sim \p^*}[\Pi_i(\p, (m_i, \vm_{-i}), \om) ] < \E_{\omega \sim \p^*}[\Pi_i(\p, (M_i,
\vm_{-i}), \om) ] 
$
or
$ 
0 > \E_{\omega \sim \p^*}[\Pi_i(\p, (m_i, \vm_{-i}), \om) ] > \E_{\omega \sim \p^*}[\Pi_i(\p, (M_i,
\vm_{-i}), \om) ]. 
$
\end{itemize}

Previously studied wagering mechanisms \citep{LLW+08,CDPV14,LL+15}
achieve truthfulness by incorporating \emph{strictly proper scoring
  rules} \citep{Savage:71} into their profit
functions. Scoring rules reward individuals based on the accuracy of
their predictions about 
random variables.  For a binary random variable, a scoring rule $s$
maps a prediction or report $p \in [0,1]$ and an outcome
$\omega \in \{0,1\}$ to a score.  A strictly proper scoring
rule incentivizes a risk neutral agent to report her true
belief. 

\begin{mydef}[Strictly proper scoring rule \citep{Savage:71}]
  A function $s: [0,1] \times \{0,1\} \to \reals \cup \{-\infty\}$ is
  a \emph{strictly proper scoring rule} if for all $p, q \in [0,1]$ with $p
  \neq q$, $\E_{\omega \sim p}[s(p,\om)] > \E_{\omega \sim p} [s(q,
  \om)]$.
\end{mydef}

One common example is the Brier scoring rule~\citep{Brier:50}, defined
as $s(p,\omega) = 1 - (p-\omega)^2$.  Note that for the Brier scoring
rule, $s(p,x) \in [0,1]$ for all $p$ and $\omega$.  Any strictly
proper scoring rule with a bounded range can be scaled to have range
$[0,1]$.

The WSWMs incorporate proper scoring rules, assigning each bettor a
profit based on how her score compares to the wager-weighted average
score of all bettors, as in Algorithm~\ref{alg:wswm}.  \citet{LLW+08}
showed that the set of WSWMs satisfy the seven axioms above and is the
\emph{unique} set of deterministic mechanisms that simultaneously
satisfy budget balance, anonymity, truthfulness, normality, and
sybilproofness.

\begin{algorithm}[h!]
  \begin{algorithmic}
    \STATE{Parameters: number of bettors $n$, strictly proper scoring
      rule $s$ with range in $[0,1]$}
    \STATE{Solicit reports $\p$ and wagers $\vm$}
    \STATE{Realize state $\omega$}
    \FOR{$i=1, \ldots, n$}
    \STATE{Pay bettor $i$
      \vspace{-15pt}
      \[
\Pi_i(\p, \vm, \om) = m_i \left(s(p_i, \om) - \frac{\sum_{j\in\N}
    m_j s(p_j, \om)}{\sum_{j \in \N} m_j} \right)
      \]
      \vspace{-10pt}
    }
    \ENDFOR
\end{algorithmic}
\caption{Weighted-score wagering mechanisms \citep{LLW+08}}
  \label{alg:wswm}
\end{algorithm}

\ignore{
\begin{mydef}[Weighted-score wagering mechanisms \cite{LLW+08}]\label{def.weighted}
  A \emph{weighted-score wagering mechanism} (WSWM) is a wagering
  mechanism in which the profit function $\Pi_i$ for each bettor $i$
  is defined as
\[ \Pi_i(\p, \vm, \om) = m_i \left(s(p_i, \om) - \frac{\sum_{j\in\N}
    m_j s(p_j, \om)}{\sum_{j \in \N} m_j} \right), \]
where $s$ is a strictly proper scoring rule with range in $[0,1]$.
\end{mydef}
}

\subsection{Adding privacy}

We would like our wagering mechanism to protect the privacy of each
bettor $i$, ensuring that 
the $n-1$ other bettors cannot learn too much about $i$'s report from
their own realized profits, even if they collude.  Note that paying
each agent according to an independent scoring rule would easily
achieve privacy, but would fail budget balance and sybilproofness.  We
formalize our desire to add privacy to the other good properties of
weighted score wagering mechanisms using joint differential privacy.

\begin{itemize}[leftmargin=\msize]
\item[(h)] {\bf $\epsilon$-joint differential privacy:} The vector of
  profit functions satisfies $\epsilon$-joint differential privacy,
  i.e., $\forall i \in \N$, $\forall \p \in [0,1]^n$, $\forall p'_i
  \in [0,1]$, $\forall \m \in \reals^n_+$, $\forall \omega \in
  \{0,1\}$, and $\forall \mathcal{S} \subset \reals^{n-1}_+$,
$
\Pr[\Pi_{-i}((p_i, \p_{-i}),\m,\omega) \in \mathcal{S}] \leq e^{\epsilon} \Pr[\Pi_{-i}((p'_i, \p_{-i}),\m,\omega)
\in \mathcal{S}]. 
$
\end{itemize}

This definition requires only that the report $p_i$ of each bettor $i$
be kept private, not the wager $m_i$. Private wagers would impose more
severe limitations on the mechanism, even if wagers are restricted to
lie in a bounded range; see Section~\ref{sec:privwagers} for a
discussion.  Note that if bettor $i$'s report $p_i$ is correlated with
his wager $m_i$, as might be the case for a Bayesian
agent~\citep{LL+15}, then just knowing $m_i$ could reveal information
about $p_i$.  In this case, differential privacy would guarantee that
other bettors can infer no more about $p_i$ after observing their
profits than they could from observing $m_i$ alone.  If bettors have
immutable beliefs as assumed by \citet{LLW+08}, then reports and
wagers are not correlated and $m_i$ reveals nothing about $p_i$.

\ignore{
This requirement will be satisfied whenever the bettors have immutable
beliefs about $\omega$, as all bettors will wager the maximum amount
allowed. We leave as an open question how to achieve privacy of both
reports and wagers, in settings where a bettor's wager may be
correlated with her report. \todo{We should also think of mathematical
  justification as to why private wagers is hard/impossible.}
}

Unfortunately, it is not possible to jointly obtain properties (a)--(h)
with any reasonable mechanism.  This is due to an inherent tension
between budget balance and privacy.  This is easy to see.  Budget
balance requires that a bettor $i$'s profit is the negation of the sum
of profits of the other $n-1$ bettors, i.e., $\Pi_i(\p,\vm,\om) = -
\sum_{j \neq i} \Pi_j(\p,\vm,\om)$.  Therefore, under budget balance,
the other $n-1$ bettors could always collude to learn bettor $i$'s
profit exactly.  In order to obtain privacy, it would therefore be
necessary for bettor $i$'s profit to be differentially private in her
own report, resulting in profits that are almost entirely noise.
This is formalized in the following theorem. We omit a formal proof since
it follows immediately from the argument described here.

\begin{theorem}
\label{thm:jdpandsens}
Let $\vPi$ be the vector of profit functions for any wagering mechanism
that satisfies both budget balance and $\epsilon$-joint differential
privacy for any $\epsilon > 0$.  Then for all $i \in \N$, $\Pi_i$ is
$\eps$-differentially private in bettor $i$'s report $p_i$.
\end{theorem}

\ignore{
\begin{proof}
  Assume such a mechanism exists, and bettor $i$ receives profit
  $\Pi_i(\p,\vm,\om)$.  Since the mechanism is $\eps$-jointly
  differentially private, then $\left\{ \Pi_j(\p,\vm,\om) \right\}_{j
    \neq i}$ must be $\eps$-differentially private in $p_i$ and $m_i$.
  Since the mechanism is budget-balanced, $\Pi_i(\p,\vm,\om) =
  -\sum_{j \neq i} \Pi_j(\p,\vm,\om)$.  Then bettor $i$'s profit
  $\Pi_i(\p,\vm,\om)$ is merely a post-processing of differentially
  private outputs, so it too must be $\eps$-differentially private in
  $p_i$ and $m_i$.
\end{proof}
}

Since it is unsatisfying to consider mechanisms in which a bettor's
profit is not sensitive to her own report, we require only that budget
balance hold in expectation over the randomness of the profit
function.  An operator who runs many markets may be content with such
a guarantee as it implies that he will not lose money on average.

\begin{itemize}[leftmargin=\msize]
\item[(a$'$)] {\bf Budget balance in expectation:} The operator neither
  makes a profit nor a loss in expectation, i.e., $\forall \p \in
  [0,1]^n$, $\forall \vm \in \reals^n_+$, $\forall \om \in \{0,1\}$,
$
\sum_{i =1}^n \E \left[ \Pi_i (\p, \vm, \om) \right] = 0.
$
\end{itemize}

\subsection{Private weighted score wagering mechanisms}

Motivated by the argument above, we seek a wagering mechanism that
simultaneously satisfies properties (a$'$) and (b)--(h).  Keeping
Theorem~\ref{thm:jdpandsens} in mind, we would also like the wagering
mechanism to be defined in such a way that each bettor $i$'s profit is
sensitive to her own report $p_i$.  Sensitivity is difficult to define
precisely, but loosely speaking, we would like it to be the case that
1) the magnitude of $\E \left[ \Pi_i (\p, \vm, \om) \right]$ varies
sufficiently with the choice of $p_i$, and 2) there is not too much
noise or variance in a bettor's profit, i.e., $\Pi_i (\p, \vm, \om)$
is generally not too far from $\E \left[ \Pi_i (\p, \vm, \om)
\right]$.

\ignore{
Before presenting such a mechanism, we first provide some intuition as
to why several more obvious approaches fail to yield satisfactory
mechanisms.
}

A natural first attempt would be to employ the standard Laplace
Mechanism~\citep{DR14} on top of a WSWM, adding independent Laplace
noise to each bettor's profit.  The resulting profit vector would
satisfy $\eps$-joint differential privacy, but since Laplace random
variables are unbounded, a bettor could lose more than her wager.
Adding other forms of noise does not help; to obtain differential
privacy, the noise must be unbounded \citep{DMNS06}.  Truncating a
bettor's profit to lie within a bounded range \emph{after} noise is
added could achieve privacy, but would result in a loss of
truthfulness as the bettor's expected profit would no longer be a
proper scoring rule.



\begin{algorithm}[h!]
 \caption{Private wagering mechanism}
  \label{alg:privwag}
  \begin{algorithmic}
    \STATE{Parameters: num bettors $n$, privacy param
      $\eps$, strictly proper scoring
      rule $s$ with range in $[0,1]$}
    \STATE{Fix $\alpha = 1 - e^{-\eps}$ and $\beta = e^{-\eps}$}
    \STATE{Solicit reports $\p$ and wagers $\vm$}
    \STATE{Realize state $\omega$}
    \FOR{$i=1, \ldots, n$}
    \STATE{Independently draw random variable $x_i(p_i, \omega)$ such that
      \[
      x_i(p_i, \omega) = \begin{cases}
        1  & \mbox{w.p. }\frac{\alpha s(p_i, \omega) + \beta}{1+\beta} \\ 
        - \beta & \mbox{w.p. } \frac{1 - \alpha s(p_i, \omega)}{1+\beta}
      \end{cases}
    \]
    \vspace{-5pt}
  }
    \ENDFOR
    \FOR{$i=1, \ldots, n$}    
    \STATE{Pay bettor $i$ 
      \vspace{-10pt}
      \[
      \Pi_i(\p, \vm, \om) = m_i \left(
        \alpha s(p_i, \om) - \frac{\sum_{j \in \N} m_j x_j(p_j,
          \om)}{\sum_{j \in \N} m_j } \right)
      \]
      \vspace{-10pt}
    }
    \ENDFOR
\end{algorithmic}
\end{algorithm}

Instead, we take a different approach.  Like the WSWM, our
\emph{private wagering mechanism}, formally defined in
Algorithm~\ref{alg:privwag}, rewards each bettor based on how good his
score is compared with an aggregate measure of how good bettors'
scores are on the whole.  However, this aggregate measure is now
calculated in a noisy manner.  That is, instead of comparing a
bettor's score to a weighted average of all bettors' scores, the
bettor's score is compared to a weighted average of random variables
that are equal to bettors' scores in expectation.  As a result, each
bettor's profit is, in expectation, equal to the profit she would
receive using a WSWM, scaled down by a parameter $\alpha$ to ensure
that no bettor ever loses more than her wager, as stated in the
following lemma.  The proof, which simply shows that for each $i$,
$\E[x_i(p_i, \omega)] = \alpha s(p_i, \omega)$, is in the appendix.

\begin{lemma}
  For any number of bettors $n>0$ with reports $\p \in [0,1]^n$ and
  wagers $\m \in \reals^n_+$, for any setting of the privacy parameter
  $\epsilon > 0$, for any outcome $\omega \in \{0,1\}$, the expected value of
  bettor $i$'s profit $\Pi_i(\p,\m,\omega)$ under the private wagering
  mechanism with scoring rule $s$ is equal to bettor $i$'s profit
  under a WSWM with scoring rule $\alpha s$.
\label{lem:expprofit}
\end{lemma}

Using this lemma, we show that this mechanism does indeed satisfy
joint differential privacy as well as the other desired properties.

\begin{theorem}
  The private wagering mechanism satisfies (a$'$) budget balance in
  expectation, (b) anonymity, (c) truthfulness, (d) individual
  rationality, (e) normality, (f) sybilproofness, (g) monotonicity,
  and (h) $\epsilon$-joint differential privacy.
\end{theorem}
\begin{proof}
  Any WSWM satisfies budget balance in expectation (by satisfying
  budget balance), truthfulness, individual rationality, normality,
  sybilproofness, and monotonicity~\cite{LLW+08}.  Since these
  properties are defined in terms of expected profit,
  Lemma~\ref{lem:expprofit} implies that the private wagering
  mechanism satisfies them too.

  Anonymity is easily observed since profits are defined
  symmetrically for all bettors.

  Finally we show $\epsilon$-joint differential privacy.  We first
  prove that each random variable $x_i(p_i, \om)$ is
  $\eps$-differentially private in bettor $i$'s report $p_i$ which
  implies that the noisy aggregate of scores is private in all
  bettors' reports.  We then apply the billboard lemma (see
  Section~\ref{sec:privacyprelims}) to show that the profit vector
  $\vPi$ satisfies joint differential privacy.

  To show that $x_i(p_i, \om)$ is differentially private in $p_i$, for
  each of the two values that $x_i(p_i, \om)$ can take on we must
  ensure that the ratio of the probability it takes this value under
  any report $p$ and the probability it takes this value under any
  alternative report $p'$ is bounded by $e^\epsilon$.  Fix any $\omega
  \in \{0,1\}$.  Since $s$ has range in
  $[0,1]$,
\[ \frac{\Pr(x_i(p, \om) = 1)}{\Pr(x_i(p', \om) = 1)} = \frac{\alpha
  s(p, \om) + \beta}{\alpha s(p', \om) + \beta} \leq \frac{\alpha +
  \beta}{\beta} = \frac{1 - e^{-\eps} + e^{-\eps}}{e^{-\eps}} =
e^{\eps}, \]
\[ \frac{\Pr(x_i(p, \om) = -\beta)}{\Pr(x_i(p', \om) = -\beta)} =
\frac{1 - \alpha s(p, \om)}{1 - \alpha s(p', \om)} \leq \frac{1}{1 -
  \alpha} = \frac{1}{1 - (1 - e^{-\eps})} = e^{\eps}. \]
Thus $x_i(p_i, \om)$ is $\epsilon$-differentially private in $p_i$.
By Theorem 4 of \citet{McS09}, the vector
$(x_1(p_1, \om), \ldots, x_n(p_n, \om))$ (and thus any function of
this vector) is $\eps$-differentially private in the vector $\vp$,
since each $x_i(p_i, \om)$ does not depend on the reports of anyone
but $i$.  Since we view the wagers $m_i$ as constants,
the quantity $\sum_{j \in \N} m_j x_j(p_j, \om) / \sum_{j \in \N} m_j$
is also $\epsilon$-differentially private in the reports $\vp$.  Call
this quantity $X$.

To apply the billboard lemma, we can imagine the operator publicly
announcing the quantity $X$ to the bettors.  
Given access to $X$, each bettor is able to calculate her own profit
$\Pi_i(\p,\m,\omega)$ using only her own input and the values $\alpha$
and $\omega$.
The billboard lemma implies that the vector of profits is
$\epsilon$-joint differentially private.
\end{proof}

\subsubsection{Sensitivity of the mechanism}

Having established that our mechanism satisfies properties (a$'$) and
(b)--(h), we next address the sensitivity of the mechanism in terms of
the two facets described above: range of achievable expected profits
and the amount of noise in the profit function. This discussion sheds
light on how to set $\epsilon$ in practice.

The first facet is quantified by Lemma~\ref{lem:expprofit}.  As
$\alpha$ grows, the magnitude of bettors' expected profits grows, and
the range of expected profits grows as well.  When $\alpha$ approaches
1, the range of expected profits achievable through the private
wagering mechanism approaches that of a standard WSWM with the same
proper scoring rule.

Unfortunately, since $\alpha = 1-e^{-\epsilon}$, larger values of
$\alpha$ imply larger values of the privacy parameter $\epsilon$.
This gives us a clear tradeoff between privacy and magnitude of
expected payments.  Luckily, in practice, it is probably unnecessary
for $\epsilon$ to be very small for most markets.  A relatively large
value of $\epsilon$ can still give bettors plausible deniability.  For
example, setting $\epsilon = 1$ implies that a bettor's report can
only change the probability of another bettor receiving a particular
profit by a factor of roughly $2.7$ and leads to
$\alpha \approx 0.63$, a tradeoff that may be considered acceptable in
practice.

The second facet is quantified in the following theorem, which states
that as more money is wagered by more bettors, each bettor's realized
profit approaches its expectation.  The bound depends on
$\| \vm \|_2 / \| \vm \|_1$.  If all wagers are equal, this quantity
is equal to $1/\sqrt{n}$ and bettors' profits approach their
expectations as $n$ grows.  This is not the case at the other extreme,
when there are a small number of bettors with wagers much larger than
the rest.  The proof, which uses Hoeffding's inequality to bound the
difference between the quantity $m_j x_j(p_j, \omega)$ and its
expectation,
is in the appendix.

\begin{theorem}\label{thm:conc}
  For any $\delta \in [0,1]$, any $\epsilon > 0$, any number of
  bettors $n > 0$, any vectors of reports $\p \in [0,1]^n$ and wagers
  $\vm \in \reals^n_+$, with probability at least $1-\delta$, for all
  $i \in \N$, the profit $\Pi_i$ output by the private wagering
  mechanism satisfies
\[ \left| \Pi_i(\p, \vm, \om) - \E[\Pi_i(\p, \vm, \om)] \right| 
\leq m_i \left( \frac{\| \vm \|_2}{ \| \vm \|_1} (1 + \beta) \sqrt{\frac{ \ln{(2/\delta)}}{2}} \right). \]
\end{theorem}

The following corollary shows that if all wagers are bounded in some
range $[L,U]$, profits approach their expectations as the number of
bettors grows.

\begin{corollary}
  Fix any $L$ and $U$, $0 < L < U$.  For any $\delta \in [0,1]$, any
  $\epsilon > 0$, any $n > 0$, any vectors of reports $\p \in [0,1]^n$
  and wagers $\vm \in [L,U]^n$, with probability at least $1-\delta$,
  for all $i \in \N$, the profit $\Pi_i$ output by the private
  wagering mechanism satisfies
\[ 
\left| \Pi_i(\p, \vm, \om) - \E[\Pi_i(\p, \vm, \om)] \right| \leq m_i
\left( \frac{U}{\sqrt{n} L} (1 + \beta) \sqrt{\frac{
      \ln{(2/\delta)}}{2}} \right).
\]
\end{corollary}

\ignore{
\begin{proof}
Observe: $\frac{\| \vm \|_2}{ \| \vm \|_1} \leq \frac{\sqrt{n U^2}}{n L} \leq \frac{U}{\sqrt{n} L}$.  Then the bound follows immediately from Theorem \ref{thm:conc}.
\end{proof}
}

%

\subsubsection{Keeping wagers private}
\label{sec:privwagers}

Property (h) requires that bettors' reports be kept private but does
not guarantee private wagers.  The same tricks used in our private
wagering mechanism could be applied to obtain a privacy guarantee for
both reports and wagers if wagers are restricted to lie in a bounded
range $[L, U]$, but this would come with a great loss in sensitivity.
Under the most straightforward extension, the parameter $\alpha$
would need to be set to $(L/U) (1-e^{-\epsilon/n})$ rather than
$(1-e^{-\epsilon})$, greatly reducing the scale of achievable profits
and thus making the mechanism impractical in most settings.

Loosely speaking, the extra factor of $L/U$ stems from the fact that a
bettor's effect on the profit of any other bettor must be roughly the
same whether he wagers the maximum amount or the minimum.  The poor
dependence on $n$ is slightly more subtle. We
created a private-belief mechanism by replacing each bettor $j$'s score
$s(p_j, \omega)$ in the WSWM with a random variable $x_j(p_j, \omega)$
that is $\epsilon$-differentially private in $p_j$.  To obtain private
wagers, we would instead need to replace the full term
$m_j s(p_j, \omega) / \sum_{k \in N} m_k$ with a random variable for
each $j$.  This term depends on the wagers of \emph{all} $n$ bettors
in addition to $p_j$.  Since each bettor's profit would depend on $n$
such random variables, achieving $\epsilon$-joint differential privacy
would require that each random variable be $\epsilon/n$-differentially
private in each bettor's wager.

We believe that sacrifices in sensitivity are unavoidable and not
merely an artifact of our techniques and analysis, but leave a formal
lower bound to future work.

\section{Limits of Privacy with Cost-Function Market Makers} 
\label{sec:costfuncs}

In practice, prediction markets are often run using dynamic mechanisms
that update in real time as new information surfaces. We now turn to
the problem of adding privacy guarantees to continuous-trade
markets. We focus our attention on cost-function prediction markets,
in which all trades are placed through an automated market maker
\citep{H03,CP07,ACV13}.  The market maker can be viewed as a streaming
algorithm that takes as input a stream of trades and outputs a
corresponding stream of market states from which trade prices can be
computed.  Therefore, the privacy guarantees we seek are in the form
of Definition~\ref{def:streamdp}.  We ask whether it is possible for
the automated market maker to price trades according to a cost
function while maintaining $\epsilon(t)$-differential privacy without
opening up the opportunity for traders to earn unbounded profits,
leading the market maker to experience unbounded loss.  We show a
mostly negative result: to achieve bounded loss, the privacy term
$e^{\epsilon(t)}$ must grow faster than linearly in $t$, the number of
rounds of trade.

For simplicity, we state our results for markets over a single binary
security, though we believe they extend to cost-function markets over
arbitrary security spaces.

\subsection{Standard cost-function market makers}

We consider a setting in which there is a single binary security that
traders may buy or sell. After the outcome $\omega \in \{0,1\}$ has
been revealed, a share of the security is worth \$1 if $\omega = 1$
and \$0 otherwise. A cost-function prediction market for this security
is fully specified by a convex function $C$ called the \emph{cost
  function}. Let $x_t$ be the number of shares that are bought or sold
by a trader in the $t$th transaction; positive values of $x_t$
represent purchases while negative values represent (short) sales. The
market state after the first $t-1$ trades is summarized by a single
value $q_t = \sum_{\tau=1}^{t-1} x_\tau$, and the $t$th trader is
charged $C(q_t + x_t) - C(q_t) = C(q_{t+1}) - C(q_{t})$.  Thus the
cost function can be viewed as a potential function, with
$C(q_{t+1}) - C(0)$ capturing the amount of money that the market
maker has collected from the first $t$ trades. The \emph{instantaneous
  price} at round $t$, denoted $p_t$, is the price per share of
purchasing an infinitesimally small quantity of shares:
$p_t = C'(q_t)$.
This framework is summarized in Algorithm~\ref{alg:cost}.

\ignore{
Algorithm \ref{alg:cost} shows first the standard cost-function
prediction market. At every time $t$, there is a current \emph{market
  state} $q_t = \sum_{i=1}^{t-1} x_i$, which is the quantity of shares
in the market when the $t$-th trader arrives at the market. We assume
the market maker sets prices according to some convex \emph{cost
  function} $C(\cdot)$; the cost of making trade $x_t$ when the market
state is $q_t$ is: $C(q_t + x_t) - C(q_t)$. The instantaneous
\emph{price} $p_t$ at state $q_t$ is the cost of making an
infinitesimally small trade: $p_t = C'(q_t)$.
} 

\begin{algorithm}[h!]
  \caption{Cost-function market maker (parameters: cost function $C$)}\label{alg:cost}
  \begin{algorithmic}
  \STATE{\textbf{Initialize:} $q_1=0$}
    \FOR{$t=1, 2, \ldots$}
  \STATE{Update instantaneous price $p_t = C'(q_t)$ }
\STATE{A trader buys $x_t \in \reals$ shares and pays $C(q_t + x_t) - C(q_t)$}
 \STATE{Update market state $q_{t+1} = q_t + x_t$}
  \ENDFOR
  \STATE{Realize outcome $\om$}
  \IF{$\om = 1$}
  \FOR{$t=1, 2, \ldots$}
  \STATE{Market maker pays $x_t$ to the trader from round $t$}
  \ENDFOR
  \ENDIF
\end{algorithmic}
\end{algorithm}

The most common cost-function market maker is Hanson's log market
scoring rule (LMSR) \citep{H03}. The cost function for the
single-security version of LMSR can be written as
$C(q) = b \log(e^{(q+a)/b} + 1)$ where $b > 0$ is a parameter
controlling the rate at which prices change as trades are made and $a$
controls the initial market price at state $q=0$. The instantaneous
price at any state $q$ is $C'(q) = e^{(q+a)/b}/(e^{(q+a)/b} + 1)$.

Under mild conditions on $C$, all cost-function market makers satisfy
several desirable properties, including natural notions of
no-arbitrage and information incorporation~\citep{ACV13}.  We refer to
any cost function $C$ satisfying these mild conditions as a
\emph{standard cost function}.
Although the market maker subsidizes trade, crucially its worst-case
loss is bounded.  This ensures that the market maker does not go
bankrupt, even if traders are perfectly informed.  Formally, there
exists a finite bound $B$ such that for any $T$, any sequence of
trades $x_1, \ldots, x_T$, and any outcome $\omega \in \{0,1\}$,
\[
q_{T+1} \cdot \one(\omega = 1) - (C(q_{T+1}) - C(0)) \leq B, 
\]
where $\one$ is the indicator function that is 1 if its argument is
true and 0 otherwise. The first term on the left-hand side is the
amount that the market maker must pay (or collect from) traders when
$\omega$ is revealed. The second is the amount collected from
traders. For the LMSR with initial price $p_1 = 0.5$ ($a = 0$), the
worst-case loss is $b \log(2)$.

\subsection{The noisy cost-function market maker}

Clearly the standard cost-function market maker does not ensure
differential privacy. The amount that a trader pays is a function of
the market state, the sum of all past trades. Thus anyone observing
the stream of market prices could infer the exact sequence of past
trades.
To guarantee privacy while still approximating cost-function pricing,
the marker maker would need to modify the sequence of published prices
(or equivalently, market states) to ensure that such information
leakage does not occur.

In this section, we define and analyze a \emph{noisy} cost-function
market maker. The noisy market maker prices trades according to a cost
function, but uses a noisy version of the market state in order to
mask the effect of past trades. In particular, the market maker
maintains a noisy market state $q_t' = q_t + \eta_t$, where $q_t$ is
the true sum of trades and $\eta_t$ is a (random) noise term. The cost
of trade $x_t$ is $C(q_t' + x_t) - C(q_t')$, with the instantaneous
price now $p_t = C'(q_t')$.  Since the noise term $\eta_t$ must be
large enough to mask the trade $x_t$, we limit trades to be some
maximum size $k$. A trader who would like to buy or sell more than $k$
shares must do this over multiple rounds.  The full modified framework
is shown in Algorithm~\ref{alg:costpriv}.  For now we allow the noise
distribution $\dist$ to depend arbitrarily on the history of trade.
This framework is general; the natural adaptation of the
privacy-preserving data market of \citet{WFA15} to the single security
prediction market setting would result in a market maker of this form,
as would a cost-function market that used existing private streaming
techniques for bit counting~\citep{CSS11, DNPR10} to keep noisy,
private counts of trades.

\ignore{
We modify this setting in Algorithm \ref{alg:costpriv} to ensure
differential privacy by maintaining a noisy version of the market
state, $q_t' = q_t + \eta_t$, where $\eta_t$ is a (random) noise term.
Prices are computed based on this noisy market state: the cost of
trade $x_t$ is $C(q_t' + x_t) - C(q_t')$, and the instantaneous price
is $p_t = C'(q_t')$. After a trade, the true quantity will be updated
to $q_{t+1} = q_t + x_t$, the noisy quantity will be updated to
$q_{t+1}' = q_{t+1} + \eta_{t+1}$ with a fresh noise term
$\eta_{t+1}$. \todo{Update the next sentence based on the final proof
  once it's finished} The distribution of $\eta_t$ can depend on the
time step $t$, the most recent market state $q_t$ and trade and $x_t$,
and on realizations of past noise terms, $\eta_1, \ldots, \eta_{t-1}$.
} 

\begin{algorithm}[h!]
  \caption{Noisy cost-function market maker (parameters: cost function $C$, distribution $\dist$
    over noise $\{ \eta_t \}$, maximum trade size $k$)}\label{alg:costpriv}
  \begin{algorithmic}
  \STATE{\textbf{Initialize:} $q_1=0$} 
  \STATE{Draw $\eta_1$ and set $q_1' = \eta_1$}
    \FOR{$t=1, 2, \ldots$}
 \STATE{Update instantaneous price $p_t = C'(q_t')$ }
 \STATE{A trader buys $x_t \in [-k, k]$ shares and pays $C(q_t' + x_t) - C(q_t')$}
 \STATE{Update true market state $q_{t+1} = q_t + x_t$}
 \STATE{Draw $\eta_{t+1}$ and update noisy market state $q_{t+1}' = q_{t+1} + \eta_{t+1}$}
 \ENDFOR
  \STATE{Realize outcome $\om$}
  \IF{$\om = 1$}
  \FOR{$t=1, 2, \ldots$}
  \STATE{Market maker pays $x_t$ to the trader from round $t$}
  \ENDFOR
  \ENDIF
\end{algorithmic}
\end{algorithm}


In this framework, we can interpret the market maker as implementing a
noise trader in a standard cost-function market. Under this
interpretation, after a (real) trader purchases $x_t$ shares at state
$q'_t$, the market state momentarily moves to
$q'_t + x_t = q_t + \eta_t + x_t = q_{t+1} + \eta_t$. The market
maker, acting as a noise trader, then effectively ``purchases''
$\eta_{t+1} - \eta_t$ shares at this state for a cost of
$C((q_{t+1} + \eta_t) + (\eta_{t+1} - \eta_t)) - C( q_{t+1} + \eta_t)
= C(q_{t+1} + \eta_{t+1}) - C(q_{t+1} + \eta_t)$,
bringing the market state to $q_{t+1} + \eta_{t+1} = q'_{t+1}$. The
market maker makes this trade regardless of the impact on its own
loss. These noise trades obscure the trades made by real traders,
opening up the possibility of privacy.

However, these noisy trades also open up the opportunity for traders
to profit off of the noise. For the market to be practical, it is
therefore important to ensure that the property of bounded worst-case
loss is maintained. For the noisy cost-function market maker, for any
sequence of $T$ trades $x_1, \ldots, x_T$, any outcome $\omega \in
\{0,1\}$, and any \emph{fixed} noise values $\eta_1, \ldots, \eta_T$,
the loss of the market maker is
\[
L_T(x_1, \ldots, x_T, \eta_1, \ldots, \eta_T, \omega) \equiv
q_{T+1} \cdot \one(\omega = 1) - \sum_{t =1}^T \left( C(q'_t +
  x_t) - C(q'_t)  \right).
\]
As before, the first term is the (possibly negative) amount that the
market maker pays to traders when $\omega$ is revealed, and the second
is the amount collected from traders (which no longer telescopes).
Unfortunately, we cannot expect this loss to be bounded for \emph{any}
noise values; the market maker could always get extremely unlucky and
draw noise values that traders can exploit.  Instead, we consider a
relaxed version of bounded loss which holds in expectation with
respect to the noise values $\eta_t$.

In addition to this relaxation, one more modification is necessary.
Note that traders can (and should) base their actions on the current
market price. Therefore, if our loss guarantee only holds in
expectation with respect to noise values $\eta_t$, then it is no
longer sufficient to give a guarantee that is worst case over any
sequences of trades. Instead, we allow the sequence of trades to
depend on the realized noise, introducing a game between traders and
the market maker. To formalize this, we imagine allowing an adversary
to control the traders. We define the notion of a \emph{strategy} for
this adversary.

\begin{mydef}[Trader strategy]\label{def.strategy}
A \emph{trader strategy} $\s$ is a set of (possibly randomized)
functions $\s = \{s_1, s_2, \ldots\}$, with each $s_t$ mapping a
history of trades and noisy market states $(x_1, \ldots,
x_{t-1}, q'_1, \ldots, q'_t)$ to a new trade $x_t$ for the
trader at round $t$.
\end{mydef}

\ignore{
  A trader strategy $\s$ is a collection of (possibly
  randomized) functions $\s = (s_1, s_2, \ldots)$, each of the form
  $s_t: \mathbb{R}^{t-1} \times (\mathbb{R} \cup \bot)^{t-1} \to
  \mathbb{R}$. A trader following strategy $\s$ at time $t$ would
  buy: $x_t = s_t \left( (\cM_1(x), \ldots, \cM_{t-1}(x) ),
    (y_1, \ldots, y_{t-1}) \right)$ shares when the mechanism had
  previously output $(\cM_1(x), \ldots, \cM_{t-1}(x) )$, and where
  $y_i = x_i$ if this trader made the trade at time $t$, and $y_i =
  \bot$ otherwise. 
} 

Let $S$ be the set of all strategies. With this definition
in place, we can formally define what it means for a noisy
cost-function market maker to have bounded loss.

\begin{mydef}[Bounded loss for a noisy cost-function market
  maker]\label{def.unbounded}
  A noisy cost-function market maker with cost function $C$ and
  distribution $\dist$ over noise values $\eta_1, \eta_2, \ldots$ is
  said to have \emph{bounded loss} if there exists a finite $B$ such
  that for all strategies $\s \in S$, all times $T \geq 1$,
  and all $\omega \in \{0,1\}$,
\[
\E
\left[L_T(x_1, \ldots, x_T, \eta_1, \ldots, \eta_T, \omega)\right]
\leq B,
\]
where the expectation is taken over the market's noise values
$\eta_1, \eta_2, \ldots$ distributed according to $\dist$ and the
(possibly randomized) actions $x_1, x_2, \ldots$ of a trader employing
strategy $\s$.  In this case, the loss of the market maker is said to
be \emph{bounded by} $B$.  The noisy cost-function market maker has
\emph{unbounded loss} if no such $B$ exists.
\end{mydef}

If the noise values were deterministic, this definition of worst-case
loss would correspond to the usual one, but
because traders react intelligently to the specific realization of
noise, we must define worst-case loss in game-theoretic terms.

\ignore{
\begin{mydef}[Unbounded loss for a noisy cost-function market
  maker]\label{def.unbounded}
  A noisy cost-function market maker with cost function $C$ and
  distribution $\dist$ over noise values $\eta_1, \eta_2, \ldots$ is
  said to have \emph{unbounded loss} if for any $B > 0$, there exists
  a strategy $\s \in \mathcal{S}$, time $T \geq 1$, and $\omega \in
  \{0,1\}$, such that
\[
\E_{\vx \sim \s, \; \eta \sim \dist}
\left[q_{T+1} \cdot \one(\omega = 1) - \sum_{t =1}^T \left( C(q_t + \eta_t +
  x_t) - C(q_t + \eta_t)  \right)\right] > B.
\]
If for some $B$ this fails to hold, the noisy cost-function market
maker has \emph{bounded loss} with loss bound $B$.
\end{mydef}
} 

\ignore{
\begin{mydef}[Worst-case loss of a noisy cost-function market
  maker]\label{def.wcl}
The worst-case loss of a noisy cost-function market
  maker with cost function $C$ and distribution $\dist$ over noise
  values $\eta_1, \eta_2, \ldots$ is 
\[
\sup_{\omega \in \{0,1\}} 
\sup_{\s \in \mathcal{S}}
\E_{\vx \sim \s, \; \eta \sim \dist, \; T}
\left[q_{T+1} \cdot \one(\omega = 1) - \sum_{t =1}^T \left( C(q_t + \eta_t +
  x_t) - C(q_t + \eta_t)  \right)\right].
\]
\end{mydef}
} 

\ignore{ 
  A trader's profit from following strategy $\s$ for $T$ rounds, when
  noise terms $\{ \eta_t\}$ are drawn from according to $G$ \rc{This
    is abuse of notation because the sequence of $x$ and $\eta$ can
    depend on each other} is:
\[ \E[\Prof((\s, T), \eta, \om)] = \E_{\vx \sim \s, \; \eta \sim G}
\left[ \Prof(\vx, \eta, \om) \right] = \mathbbm{1}_{\om=1}
\sum_{t=1}^T \E[x_t] - \sum_{t =1}^T \E \left[ C(q_t + \eta_t + x_t) -
  C(q_t + \eta_t) \right] \]
} 

\subsection{Limitations on privacy}

By effectively acting as a noise trader, a noisy cost-function market
maker can partially obscure trades. Unfortunately, the amount of
privacy achievable through this technique is limited. In this section,
we show that in order to simultaneously maintain bounded loss and
achieve $\epsilon(t)$-differential privacy, the quantity
$e^{\epsilon(t)}$ must grow faster than linearly as a function of the
number of rounds of trade.

Before stating our result, we explain how to frame the market maker
setup in the language of differential privacy. Recall from
Section~\ref{sec:privstreams} that a differentially private unbounded
streaming algorithm $\cM$ takes as input a stream $\sigma$ of
arbitrary length and outputs a stream of values that depend on
$\sigma$ in a differentially private way. In the market setting, the
stream $\sigma$ corresponds to the sequence of trades
$\x = (x_1, x_2, \ldots )$.  We think of the noisy cost-function
market maker (Algorithm \ref{alg:costpriv}) as an algorithm $\cM$
that, on any stream prefix $(x_1, \ldots, x_t)$, outputs the noisy
market states $(q'_1, \ldots, q'_{t+1})$.\footnote{Announcing $q'_t$
  allows traders to infer the instantaneous price $p_t = C'(q'_t)$. It
  is equivalent to announcing $p_t$ in terms of information revealed
  as long as $C$ is strictly convex in the region around $q'_t$.}
The goal is to find a market maker such that $\cM$ is
$\epsilon(t)$-differentially private.

One might ask whether it is necessary to allow the privacy guarantee
to diminish as the the number of trades grows.  When considering the
problem of calculating noisy sums of bit streams, for example,
\citet{CSS11} are able to maintain a fixed privacy guarantee as their
stream grows in length by instead allowing the accuracy of their
counts to diminish.  This approach doesn't work for us; we cannot
achieve bounded loss yet allow the market maker's loss to grow with
the number of trades.

Our result relies on one mild assumption on the distribution $\dist$
over noise.  In particular, we require that the noise $\eta_{t+1}$ be
chosen independent of the current trade $x_t$.~\footnote{The proof can
  be extended easily to the more general case in which the calculation
  of $\eta_{t+1}$ is differentially private in $x_t$; we make the
  slightly stronger assumption to simplify presentation.}  We refer to
this as the \emph{trade-independent noise assumption}.  The
distribution of $\eta_{t+1}$ may still depend on the round $t$, the
history of trade $x_1, \ldots, x_{t-1}$, and the realizations of past
noise terms, $\eta_1, \ldots, \eta_{t}$. 
%
This assumption is needed in the proof only to rule out unrealistic
market makers that are specifically designed to monitor and infer the
behavior of the specific adversarial trader that we consider, and the
result likely holds even without it.  However, it is not a terribly
restrictive assumption as most standard ways of generating noise could
be written in this form.  For example, \citet{CSS11} and
\citet{DNPR10} show how to maintain a noisy count of the number of
ones in a stream of bits.
Both achieve this by computing the exact count and adding noise that
is correlated across time but independent of the data.  If similar
ideas were used to choose the noise term in our setting, the
trade-independent noise assumption would be satisfied.
\ignore{
The foundational tool for achieving differential privacy in a static
(non-streaming) setting is the Laplace Mechanism of \citet{DMNS06},
which adds independent noise to the value of a single statistic.  Both
\citet{CSS11} and \citet{DNPR10} give a natural extension this idea to
a streaming setting by adding noise terms that are independent of the
stream at each time $t$.  If either algorithm was used to determine
the noise terms $\eta_t$ within a noisy cost-function market maker,
the market would satisfy the trade-independent noise assumption.  
}
The noise employed in the mechanism of \citet{WFA15} also satisfies
this assumption.  Our impossibility result then implies that their
market would have unbounded loss if a limit on the number of rounds of
trade were not imposed. To obtain privacy guarantees,
\citeauthor{WFA15} must assume that the number of trades is known in
advance and can therefore be used to set relevant market parameters.

\ignore{
In addition to satisfying differential privacy, one would hope the mechanism $\cM$ also provided output that was relevant and informative about the database $x$.   We observe here that if a trader is allowed to make arbitrarily large trades (i.e. $x_t$ unbounded), then this task is hopeless.  Intuitively, the $\eta_t$ terms must be large enough to mask a single trade $x_t$, but if $x_t$ can be unbounded, then the added noise must also be unbounded most of the time.  Differential privacy guarantees that the noisy output $\cM_{t+1}(x)$ must be similar for any trade $x_t \in \mathbb{R}$, so the noisy state cannot be within a finite radius of the true state with any positive probability.  \rc{We have a formal statement of this that has been commented out}  The main negative result in the next section considers players who can only trade a bounded number of shares per trade (i.e. $x_t$ bounded).  Based on the above observation, this restriction only strengthens the impossibility result.  \rc{I would like to eventually prove this claim for formally.  The way it's written now is too hand-wavy.  Maybe it will naturally arise in the proof once it's finished.}
} 

We now state the main result.

\begin{theorem}
  Consider any noisy cost-function market maker using a standard
  convex cost function $C$ that is nonlinear in some region, a noise
  distribution $\dist$ satisfying the trade-independent noise
  assumption, and a bound $k > 0$ on trade size.  If the market maker
  satisfies bounded loss, then it cannot satisfy
  $(\epsilon(t), \delta)$-differential privacy for any function
  $\epsilon$ such that $e^{\epsilon(t)} = O(t)$ with any
  constant $\delta \in [0,1)$.
\label{thm:unboundedloss}
\end{theorem}

This theorem rules out bounded loss with $\epsilon(t) = \log(m t)$ for
any constant $m > 0$.  It is open whether it is possible to achieve
$\epsilon(t) = m \log(t)$ (and therefore $e^{\epsilon(t)} = t^m$) for
some $m > 1$, but such a guarantee would likely be insufficient in
most practical settings.

Note that with unbounded trade size (i.e., $k = \infty$), our result
would be trivial. A trader could change the market state (and hence
the price) by an arbitrary amount in a single trade. To provide
differential privacy, the noisy market state would then have to be
independent of past trades.  The noisy market price would not be
reflective of trader beliefs, and the noise added could be easily
exploited by traders to improve their profits.  By imposing a bound on
trade size, we only strengthen our negative result.


\ignore{
\begin{thm}
For any $\eps$-differentially private cost function based prediction market, where $\eps$ is a constant and each noise term $\eta_t$ is conditionally independent of the most recent trade $x_t$, given the history of play, the market maker can experience unbounded loss.
\end{thm}
} 


While the proof of Theorem~\ref{thm:unboundedloss} is quite technical,
the intuition is simple.  We consider the behavior of the noisy
cost-function market maker when there is a single trader trading in
the market repeatedly using a simple trading strategy.  This trader
chooses a \emph{target state} $q^*$.  Whenever the noisy market state
$q'_t$ is less than $q^*$ (and so $p_t < p^* \equiv C'(q^*)$), the
trader purchases shares, pushing the market state as close to $q^*$ as
possible.  When the noisy state $q'_t$ is greater than $q^*$ (so
$p_t > p^*$), the trader sells shares, again pushing the state as
close as possible to $q^*$.  Each trade makes a profit for the trader
\emph{in expectation} if it were the case that $\omega = 1$ with
probability $p^*$.  Since there is only a single trader, this means
that each such trade would result in an expected loss with respect to $p^*$ for the market maker. Unbounded expected loss for any $p^*$ implies unbounded loss in the worst case---either when $\omega=0$ or $\omega=1$.  The crux of
the proof involves showing that in order achieve bounded loss against
this trader, the amount of added noise $\eta_t$ cannot be too big as
$t$ grows, resulting in a sacrifice of privacy.

To formalize this intuition, we first give a more precise description
of the strategy $\s^*$ employed by the single trader we consider.

\begin{mydef}[Target strategy]\label{def.target}
  The \emph{target strategy} $\s^*$ with target $q^* \in \reals$
  chosen from a region in which $C$ is nonlinear is defined as
  follows.  For all rounds $t$,
\[
s^*_t(x_1, \ldots, x_{t-1}, q'_1, \ldots, q'_t)
=  \begin{cases}
      \min\{ q^* - q'_t, k\}  & \textrm{if $q'_t \leq q^*$}, \\
      - \min\{ q'_t - q^*, k\} & \textrm{otherwise}.
    \end{cases}
\]
\end{mydef}

As described above, if $\omega=1$ with probability $p^*$, a trader
following this target strategy makes a non-negative expected profit on
every round of trade.  Furthermore, this trader makes an expected
profit of at least some constant $\chi > 0$ on each round in which the
noisy market state $q'_t$ is more than a constant distance $\gamma$
from $q^*$.  The market maker must subsidize this profit, taking an
expected loss with respect to $p^*$ on each round.  These ideas are
formalized in Lemma~\ref{lem:profit}, which lower bounds the expected
loss of the market maker in terms of the probability of $q'_t$ falling
far from $q^*$.  In this statement, $D_C$ denotes the Bregman
divergence\footnote{The \emph{Bregman divergence} of a convex function
  $F$ of a single variable is defined as
  $D_F(p,q) = F(p) - F(q) - F'(q)(p-q)$.  The Bregman divergence is
  always non-negative. If $F$ is strictly convex, it is strictly
  positive when the arguments are not equal.\label{fn:breg}} of
$C$. The proof is in the appendix.

\begin{lemma}
  Consider a noisy cost-function market maker satisfying the
  conditions in Theorem~\ref{thm:unboundedloss} with a single trader
  following the target strategy $\s^*$ with target $q^*$.  Suppose
  $\omega = 1$ with probability $p^* = C'(q^*)$.  Then for any
  $\gamma$ such that $0 < \gamma \leq k$,
\[
\E
\left[L_T(x_1, \ldots, x_T, \eta_1, \ldots, \eta_T, \omega)\right]
\geq 
\chi \sum_{t=1}^T \Pr( |q'_t - q^* | \geq \gam) 
\]
where the expectation and probability are taken over the randomness in
the noise values $\eta_1, \eta_2, \ldots$, the resulting actions $x_1,
x_2, \ldots$ of the trader, and the random outcome $\omega$, and where
$\chi = \min\{
D_C(q^* + \gamma, q^*),
D_C(q^* - \gamma, q^*)
\} > 0$.
\label{lem:profit}
\end{lemma}

\ignore{
Intuitively, Lemma \ref{lem.profit} lower bounds the trader's (possibly complicated) per-round profit by $\chi$ if the noisy quantity is in $[q - \gam, q + \gam]$ and 0 otherwise.  The value of $\chi$ will then the minimum of her expected profit from a trade that pushes the quantity either from $q-\gam$ to $q$ or from $q+\gam$ to $q$.  If the trader buys or sells any additional shares,  it is because the noisy quantity was strictly outside of the range $[q - \gam, q + \gam]$.  We can then consider her trade in two parts: first, the $|x_t - \gam|$ shares that moved the quantity from $q_t + \eta_t$ to $q \pm \gam$, and second, the remaining $\gam$ shares that moved the quantity from exactly $q \pm c$ to $q$, from which she receives expected profit at most $\chi$.   As described in the proof of Lemma \ref{lem.profit}, her per-share profit from the first portion of the trade must be higher than per-share profit from the latter portion of the trade, so her expected profit from the entire trade is at least:
\[ \E_p\left[ \mbox{profit}(x_t) \right] \geq \frac{\chi}{\gam}(|x_t - \gam|) + \chi \geq \chi. \]
}

\ignore{
\todo{formalize the last argument in two cases, either here or in proof of Lemma \ref{lem.profit}: 1. starting at any market state $q$ such that $C'(q+\tau) <=p$ and buying tau shares, and 2. starting at any market state q such that $C'(q-\tau) >= p$ and selling tau shares} \todo{Do this the first time it's introduced}
}

We now complete the proof.

\begin{proof}[Proof of Theorem~\ref{thm:unboundedloss}]
  We will show that bounded loss implies that
  $(\epsilon(t), \delta)$-differential privacy cannot be achieved with
  $e^{\epsilon(t)} = O(t)$ for any constant $\delta \in [0,1)$.

  Throughout the proof, we reason about the probabilities of various
  events conditioned on there being a single trader playing a
  particular strategy.  
  All strategies we consider are
  deterministic, so all probabilities are taken just with respect to
  the randomness in the market maker's added noise ($\eta_1, \eta_2,
  \ldots$).

  As described above, we focus on the case in which a single
  trader plays the target strategy $\s^*$ with target $q^*$.  Define
  $R^*$ to be the open region of radius $k/4$ around $q^*$, that is,
  $R^* = (q^* - k/4, q^* + k/4)$.  Let $\hat{q} = q^* + k/2$ and let
  $\hat{R} = (\hat{q} - k/4, \hat{q} + k/4)$.  Notice that $R^*$ and
  $\hat{R}$ do not intersect, but from any market state $q \in R^*$ a
  trader could move the market state to $\hat{q}$ with a purchase or
  sale of no more than $k$ shares. 
 \begin{figure}[t!]
\begin{nscenter} 
\includegraphics[scale=0.6]{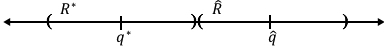}
\vspace{-10pt}
\end{nscenter}
\end{figure}

  For any round $t$, let $\s^{t}$ be the strategy in which
  $s^{t}_{\tau} = s^*_{\tau}$ for all rounds $\tau \ne t$, but
  $s^{t}_t(x_1, \ldots, x_{t-1}, q'_1, \ldots, q'_t) = \hat{q} -
  q'_t$ if $|\hat{q} - q'_t| \leq k$ (otherwise, $s^{t}_t$ can be
  defined arbitrarily).  In other words, a trader playing strategy
  $\s^{t}$ behaves identically to a trader playing strategy $\s^*$ on
  all rounds except round $t$.  On round $t$, the trader instead
  attempts to move the market state to $\hat{q}$.

  For any $t$, the behavior of a trader playing strategy $\s^*$ and a
  trader playing strategy $\s^t$ are indistinguishable through round
  $t-1$, and therefore the behavior of the market maker is
  indistinguishable as well.  At round $t$, if it is the case that
  $q'_t \in R^*$ (and therefore $|q'_t - q^*| \leq k/4 < k$ and also $|q'_t - \hat{q}| \leq 3k/4 < k$), then a
  trader playing strategy $\s^*$ would purchase $q^* - q'_t$ shares,
  while a trader playing strategy $\s^t$ would purchase $\hat{q} -
  q'_t$.  Differential privacy tells us that conditioned on such a
  state being reached, the probability that $q'_{t+1}$ lies in any
  range (and in particular, in $R^*$) should not be too different
  depending on which of the two actions the trader takes.  More
  formally, if the market maker satisfies $\epsilon(t)$-differential
  privacy, then for all rounds $t$, 
  \begin{align*}
    e^{\epsilon(t)} 
    &\geq \frac{\Pr(q'_{t+1} \in R^* | \s = \s^*, q'_t \in
      R^*) -  \delta }{\Pr(q'_{t+1} \in R^* | \s = \s^t, q'_t \in R^*)}
    \geq \frac{\Pr(q'_{t+1} \in R^* | \s = \s^*, q'_t \in
      R^*) - \delta }{\Pr(q'_{t+1} \not\in \hat{R} | \s = \s^t, q'_t \in
      R^*)}
    \\
    &= \frac{\Pr(q'_{t+1} \in R^* | \s = \s^*, q'_t \in
      R^*) - \delta }{\Pr(q'_{t+1} \not\in R^* | \s = \s^*, q'_t \in
      R^*)}.
  \end{align*}
  The first inequality follows from the definition of
  $(\epsilon(t), \delta)$-differential privacy.  The second follows from
  the fact that $R^*$ and $\hat{R}$ are disjoint.  The last line is a
  consequence of the trade-independent noise assumption.  
  By simple
  algebraic manipulation, for all $t$,
  \begin{align}
  \Pr(q'_{t+1} \not\in R^* | \s = \s^*, q'_t \in R^*)
  \geq \frac{1 - \delta }{1+e^{\epsilon(t)}}.
    \label{eqn:epsbound}
  \end{align}

  We now further investigate the term on the left-hand side of this
  equation.  For the remainder of the proof, we assume that
  $\s = \s^*$ and implicitly condition on this.  

  Applying Lemma~\ref{lem:profit} with $\gamma = k/4$, we find that
  the expected value of the market maker's loss after $T$ rounds if
  $\omega = 1$ with probability $p^* = C'(q*)$ is lower bounded by
  $\chi \sum_{t=1}^T \Pr(q'_t \not\in R^*)$ for the appropriate constant $\chi$.
  This implies that for at least one of $\omega = 1$ or $\omega = 0$, 
  $
  \E\left[L_T(x_1, \ldots, x_T, \eta_1, \ldots, \eta_T, \omega)\right]
  \geq \chi \sum_{t=1}^T \Pr(q'_t \not\in R^*)
  $
  where the expectation is just over the random noise of the market
  maker and the resulting actions of the trader.  Since we have
  assumed that the market maker's loss is bounded, this implies there
  must exist some loss bound $B$ such that 
  \begin{align}
    \frac{B}{\chi} 
    & \geq \sum_{t=1}^\infty \Pr(q'_t \not\in R^*).
    \label{eqn:Bchibound}
  \end{align}
  Fix any constant $\alpha \in (0,1)$.  Equation~\ref{eqn:Bchibound}
  implies that for all but finitely many $t$, $\Pr(q'_t \not\in R^*) <
  \alpha$, or equivalently, for all but finitely many $t$, $\Pr(q'_t
  \in R^*) \geq 1-\alpha$.  Call the set of $t$ for which this holds
  $\T$.  Equation~\ref{eqn:Bchibound} also implies that
  \begin{align*}
    \frac{B}{\chi} 
    & \geq \sum_{t=1}^\infty \left[ \Pr(q'_{t+1} \not\in R^* | q'_t \in R^*)
    \Pr(q'_t \in R^*) + \Pr(q'_{t+1} \not\in R^* | q'_t \not\in R^*)
    \Pr(q'_t \not\in R^*) \right]
    \\
    & \geq \sum_{t=1}^\infty \Pr(q'_{t+1} \not\in R^* | q'_t \in R^*)
    \Pr(q'_t \in R^*)
  \geq (1-\alpha) \sum_{t \in \T} \Pr(q'_{t+1} \not\in R^* | q'_t \in
    R^*).
  \end{align*}
 
  Combining this with Equation~\ref{eqn:epsbound} yields
  \begin{align}
  \sum_{t \in \T} \frac{1 - \delta }{1+e^{\epsilon(t)}} \leq \frac{B}{\chi (1-\alpha)}.
  \label{eqn:sumbound}
  \end{align}

  Now suppose for contradiction that $e^{\epsilon(t)} = O(t)$.  Then
  by definition, for some constant $m>1$ there exists a round
  $\tau$ such that for all $t > \tau$, $e^{\epsilon(t)} \leq m t$.
  Then
  \begin{align*}
    \sum_{t \in \T} \frac{1 - \delta }{1+e^{\epsilon(t)}} 
    &\geq \sum_{t \in \T, t > \tau} \frac{1 - \delta }{1+e^{\epsilon(t)}}
    \geq \sum_{t \in \T, t > \tau} \frac{1 - \delta }{1+mt}
    > \frac{1 - \delta }{m} \sum_{t \in \T, t > \tau} \frac{1}{1+t}.
  \end{align*}
  Since this sum is over all natural numbers $t$ except a finite
  number, it must diverge, and therefore Equation~\ref{eqn:sumbound}
  cannot hold.  Therefore, we cannot have $e^{\epsilon(t)} = O(t)$.
\end{proof}

\section{Discussion}

We designed a class of randomized wagering mechanisms that keep
bettors' reports private while maintaining truthfulness, budget
balance in expectation, and other desirable properties of weighted
score wagering mechanisms.  The parameters of our mechanisms can be
tuned to achieve a tradeoff between the level of privacy guaranteed
and the sensitivity of a bettor's payment to her own report.
Determining how to best make this tradeoff in practice (and more
generally, what level of privacy is acceptable in differentially
private algorithms) is an open empirical question.

While our results in the dynamic setting are negative, there are
several potential avenues for circumventing our lower bound.  The
lower bound shows that it is not possible to obtain reasonable privacy
guarantees using a noisy cost-function market maker when traders may
buy or sell fractional security shares, as is typically assumed in the
cost function literature.  Indeed, the adversarial trader we consider
buys and sells arbitrarily small fractions when the market state is
close to its target.  This behavior could be prevented by enforcing a
minimum unit of purchase.  Perhaps cleverly designed noise could allow
us to avoid the lower bound with this additional restriction.
However, based on preliminary simulations of a noisy cost-function
market based on Hanson's LMSR~\citeyear{H03} with noise drawn using
standard binary streaming approaches \citep{DNPR10,CSS11}, it appears
an adversary can still cause a market maker using these ``standard''
techniques to have unbounded loss by buying one unit when the noisy
market state is below the target and selling one unit when it is above.

One could also attempt to circumvent the lower bound by adding a
transaction fee for each trade that is large enough that traders
cannot profit off the market's noise.  While the fee could always be
set large enough to guarantee bounded loss, a large fee would
discourage trade in the market and limit its predictive power.  A
careful analysis would be required to ensure that the fee could be set
high enough to maintain bounded loss without rendering the market
predictions useless.

\bibliographystyle{plainnat}
\bibliography{refs}

\begin{thebibliography}{28}
\providecommand{\natexlab}[1]{#1}
\providecommand{\url}[1]{\texttt{#1}}
\expandafter\ifx\csname urlstyle\endcsname\relax
  \providecommand{\doi}[1]{doi: #1}\else
  \providecommand{\doi}{doi: \begingroup \urlstyle{rm}\Url}\fi

\bibitem[Abernethy et~al.(2013)Abernethy, Chen, and Vaughan]{ACV13}
Jacob Abernethy, Yiling Chen, and Jennifer~Wortman Vaughan.
\newblock Efficient market making via convex optimization, and a connection to
  online learning.
\newblock \emph{ACM Transactions on Economics and Computation}, 1\penalty0 (2),
  2013.

\bibitem[Berg and Proebsting(2009)]{BP2009}
Henry Berg and Todd~A. Proebsting.
\newblock Hanson's automated market maker.
\newblock \emph{Journal of Prediction Markets}, 3\penalty0 (1):\penalty0
  45--59, 2009.

\bibitem[Berg et~al.(2001)Berg, Forsythe, Nelson, and Rietz]{Ber:01}
Joyce~E. Berg, Robert Forsythe, Forrest~D. Nelson, and Thomas~A. Rietz.
\newblock Results from a dozen years of election futures markets research.
\newblock In C.~A. Plott and V.~Smith, editors, \emph{Handbook of Experimental
  Economic Results}. 2001.

\bibitem[Brier(1950)]{Brier:50}
Glenn~W. Brier.
\newblock Verification of forecasts expressed in terms of probability.
\newblock \emph{Monthly Weather Review}, 78\penalty0 (1):\penalty0 1--3, 1950.

\bibitem[Chan et~al.(2011)Chan, Shi, and Song]{CSS11}
T.-H.~Hubert Chan, Elaine Shi, and Dawn Song.
\newblock Private and continual release of statistics.
\newblock \emph{ACM Transactions on Information and System Security},
  14\penalty0 (3):\penalty0 26, 2011.

\bibitem[Charette(2007)]{C07}
Robert Charette.
\newblock An internal futures market.
\newblock \emph{Information Management}, 2007.

\bibitem[Chen and Pennock(2007)]{CP07}
Yiling Chen and David~M. Pennock.
\newblock A utility framework for bounded-loss market makers.
\newblock In \emph{Proc. of the Conference on Uncertainty in Artificial
  Intelligence}, 2007.

\bibitem[Chen et~al.(2014)Chen, Devanur, Pennock, and Vaughan]{CDPV14}
Yiling Chen, Nikhil~R. Devanur, David~M. Pennock, and Jennifer~Wortman Vaughan.
\newblock Removing arbitrage from wagering mechanisms.
\newblock In \emph{Proceedings of the 15th ACM Conference on Economics and
  Computation}, 2014.

\bibitem[Cowgill and Zitzewitz(2015)]{CZ15}
Bo~Cowgill and Eric Zitzewitz.
\newblock Corporate prediction markets: {E}vidence from google, ford, and firm
  x.
\newblock \emph{Review of Economic Studies}, 82\penalty0 (4):\penalty0
  1309--1341, 2015.

\bibitem[Dwork and Roth(2014)]{DR14}
Cynthia Dwork and Aaron Roth.
\newblock The algorithmic foundations of differential privacy.
\newblock \emph{Foundations and Trends in Theoretical Comp. Sci.}, 9\penalty0
  (34):\penalty0 211--407, 2014.

\bibitem[Dwork et~al.(2006)Dwork, McSherry, Nissim, and Smith]{DMNS06}
Cynthia Dwork, Frank McSherry, Kobbi Nissim, and Adam Smith.
\newblock Calibrating noise to sensitivity in private data analysis.
\newblock In \emph{Proceedings of the 3rd Conference on Theory of
  Cryptography}, 2006.

\bibitem[Dwork et~al.(2010)Dwork, Naor, Pitassi, and Rothblum]{DNPR10}
Cynthia Dwork, Moni Naor, Toniann Pitassi, and Guy~N. Rothblum.
\newblock Differential privacy under continual observation.
\newblock In \emph{Proceedings of the 42nd ACM Symposium on Theory of
  Computing}, 2010.

\bibitem[Gandar et~al.(1999)Gandar, Dare, Brown, and Zuber]{Gan:98}
John~M. Gandar, William~H. Dare, Craig~R. Brown, and Richard~A. Zuber.
\newblock Informed traders and price variations in the betting market for
  professional basketball games.
\newblock \emph{Journal of Finance}, LIII\penalty0 (1):\penalty0 385--401,
  1999.

\bibitem[Gneiting and Raftery(2007)]{Gneiting:07}
Tilmann Gneiting and Adrian~E. Raftery.
\newblock Strictly proper scoring rules, prediction, and estimation.
\newblock \emph{Journal of the American Statistical Association}, 102\penalty0
  (477):\penalty0 359--378, 2007.

\bibitem[Grossman(1976)]{Gro:76}
Sanford~J. Grossman.
\newblock On the efficiency of competitive stock markets where traders have
  diverse information.
\newblock \emph{The Journal of Finance}, 31\penalty0 (2):\penalty0 573--585,
  1976.

\bibitem[Hanson(2003)]{H03}
Robin Hanson.
\newblock Combinatorial information market design.
\newblock \emph{Information Systems Frontiers}, 5\penalty0 (1):\penalty0
  105--119, 2003.

\bibitem[Hsu et~al.(2014)Hsu, Huang, Roth, Roughgarden, and Wu]{HHR+14}
Justin Hsu, Zhiyi Huang, Aaron Roth, Tim Roughgarden, and Zhiwei~Steven Wu.
\newblock Private matchings and allocations.
\newblock In \emph{Proceedings of the 46th Annual ACM Symposium on Theory of
  Computing}, 2014.

\bibitem[Kearns et~al.(2014)Kearns, Pai, Roth, and Ullman]{KPRU14}
Michael Kearns, Mallesh Pai, Aaron Roth, and Jonathan Ullman.
\newblock Mechanism design in large games: Incentives and privacy.
\newblock In \emph{Proceedings of the 5th Conference on Innovations in
  Theoretical Computer Science}, 2014.

\bibitem[Lambert et~al.(2008)Lambert, Langford, Wortman, Chen, Reeves, Shoham,
  and Pennock]{LLW+08}
Nicolas~S. Lambert, John Langford, Jennifer Wortman, Yiling Chen, Daniel
  Reeves, Yoav Shoham, and David~M. Pennock.
\newblock Self-financed wagering mechanisms for forecasting.
\newblock In \emph{Proceedings of the 9th ACM Conference on Electronic
  Commerce}, 2008.

\bibitem[Lambert et~al.(2015)Lambert, Langford, Vaughan, Chen, Reeves, Shoham,
  and Pennock]{LL+15}
Nicolas~S. Lambert, John Langford, Jennifer~Wortman Vaughan, Yiling Chen,
  Daniel Reeves, Yoav Shoham, and David~M. Pennock.
\newblock An axiomatic characterization of wagering mechanisms.
\newblock \emph{Journal of Economic Theory}, 156:\penalty0 389--416, 2015.

\bibitem[McSherry(2009)]{McS09}
Frank McSherry.
\newblock Privacy integrated queries: An extensible platform for
  privacy-preserving data analysis.
\newblock In \emph{Proceedings of the 2009 ACM SIGMOD International Conference
  on Management of Data}, 2009.

\bibitem[Pennock et~al.(2002)Pennock, Lawrence, Giles, and
  Nielsen]{PenScience:01}
David~M. Pennock, Steve Lawrence, C.~Lee Giles, and Finn~A. Nielsen.
\newblock The real power of artificial markets.
\newblock \emph{Science}, 291:\penalty0 987--988, 2002.

\bibitem[Plott and Chen(2002)]{CP02}
Charles Plott and Kay-Yut Chen.
\newblock Information aggregation mechanisms: {C}oncept, design and field
  implementation.
\newblock California Institute of Technology Social Science Working Paper 1131,
  2002.

\bibitem[Polgreen et~al.(2007)Polgreen, Nelson, and Neumann]{Polgreen:2007}
Philip~M. Polgreen, Forrest~D. Nelson, and George~R. Neumann.
\newblock Using prediction markets to forecast trends in infectious diseases.
\newblock \emph{Clinical Infectious Diseases}, 44\penalty0 (2):\penalty0
  272--279, 2007.

\bibitem[Roll(1984)]{Roll:84}
Richard Roll.
\newblock Orange juice and weather.
\newblock \emph{The American Economic Review}, 74\penalty0 (5):\penalty0
  861--880, 1984.

\bibitem[Savage(1971)]{Savage:71}
Leonard~J. Savage.
\newblock Elicitation of personal probabilities and expectations.
\newblock \emph{Journal of the American Statistical Association}, 66\penalty0
  (336):\penalty0 783--801, 1971.

\bibitem[Thaler and Ziemba(1988)]{Tha:88}
Richard~H. Thaler and William~T. Ziemba.
\newblock Anomalies: Parimutuel betting markets: Racetracks and lotteries.
\newblock \emph{J. of Economic Perspectives}, 2\penalty0 (2):\penalty0
  161--174, 1988.

\bibitem[Waggoner et~al.(2015)Waggoner, Frongillo, and Abernethy]{WFA15}
Bo~Waggoner, Rafael Frongillo, and Jacob Abernethy.
\newblock A market framework for eliciting private data.
\newblock In \emph{Advances in Neural Information Processing Systems 28}, 2015.

\end{thebibliography}

\newpage
\appendix
\section*{APPENDIX}
\setcounter{section}{1}
\subsection{Omitted Proofs from Section~\ref{sec:oneshot}}

\begin{proof}[Proof of Lemma~\ref{lem:expprofit}]
For each $i \in \N$,
\begin{equation}
\E[x_i(p_i, \omega)] 
= \frac{\alpha s(p_i, \omega) + \beta}{1+\beta}
- \beta  \frac{1 - \alpha s(p_i, \omega)}{1+\beta}
= \alpha s(p_i, \omega)
\label{eqn:expx}
\end{equation}
and so
\[
   \E[\Pi_i(\p, \vm, \om)] 
   = m_i \left(\alpha s(p_i, \om) - \frac{\sum_{j \in \N} m_j \alpha s_j(p_j,
          \om)}{\sum_{j \in \N} m_j } \right).
\]
This is precisely the profit to bettor $i$ in a WSWM with scoring rule
$\alpha s$.
\end{proof}

\begin{proof}[Proof of Theorem~\ref{thm:conc}]
  For any $j \in \N$, consider the quantity $m_j x_j(p_j, \omega)$.
  From Equation~\ref{eqn:expx}, $\E[m_j x_j(p_j, \omega)] = m_j \alpha
  s(p_j, \omega)$.  Additionally we can bound $m_j x_j(p_j, \omega)
  \in [-m_j \beta, m_j]$.  Hoeffding's inequality then implies that
  with probability at least $1- \delta$,
\[
\left| \sum_{j \in \N} m_j \alpha s(p_j, \omega) - \sum_{j \in \N}
  m_j x_j(p_j, \omega) \right|
\leq \|\m\|_2 (1+\beta) \sqrt{\frac{\ln(2/\delta)}{2}}.
\]
From the definition of the private wagering mechanism and
Lemma~\ref{lem:expprofit}, we then have that with probability at least
$1-\delta$, for any $i \in \N$, 
\begin{align*}
\left| \Pi_i(\p, \vm, \om) - \E[\Pi_i(\p, \vm, \om)] \right| 
&= \frac{m_i}{\sum_{j \in \N} m_j} \left|  \sum_{j \in \N} m_j \alpha
  s(p_j, \omega) - \sum_{j \in \N} m_j x_j(p_j, \omega) \right|
\\
&\leq m_i \frac{\|\m\|_2}{\| \m \|_1} (1+\beta)
\sqrt{\frac{\ln(2/\delta)}{2}}
\end{align*}
as desired.
\end{proof}

\subsection{Omitted Proofs from Section~\ref{sec:costfuncs}}

The proof of Lemma~\ref{lem:profit} makes use of the following
technical lemma, which says that it is profitable in expectation to
sell shares as long as the price remains above $p^*$ or to purchase
shares as long as the price remains below $p^*$.  In this statement,
$q$ can be interpreted as the current market state and $x$ as a new
purchase (or sale); $C'(q^*) x - C(q+x) + C(q) \geq 0$ would then be
the expected profit of a trader making this purchase or sale if
$\omega \sim p^* = C'(q^*)$.

\begin{lemma}
  Fix any convex function $C$ and any $q^*$, $q$, and $x$ such that $q
  + x \geq q^*$ if $x \leq 0$ and $q + x \leq q^*$ if $x \geq 0$. Then
  $C'(q^*) x - C(q+x) + C(q) \geq 0$.
\label{lem:profittrade}
\end{lemma}
\begin{proof}
  Since $C$ is convex, the assumptions in the lemma statement imply
  that if $x \leq 0$ then $C'(q + x) \geq C'(q^*)$, while if $x \geq
  0$ then $C'(q + x) \leq C'(q^*)$.  Therefore, in either case
  $C'(q+x) x  \leq C'(q^*) x$, and
  \begin{align*}
    C'(q^*) x - C(q+x) + C(q)
    &\geq C'(q+x) x - C(q+x) + C(q)
   = D_C(q, q+x) \geq 0.
  \end{align*}
\end{proof}

\begin{proof}[Proof of Lemma~\ref{lem:profit}]
  From the definition of the market maker's loss, we can rewrite 
 $\E \left[L_T(x_1, \ldots, x_T, \eta_1, \ldots, \eta_T, \omega)\right] =
 \sum_{t=1}^T \E[\pi_t]$
  where $\pi_t$ is the expected (over just the
  randomness in $\omega$) loss of the market maker from the $t$th
  trade, i.e.,
  \[
  \pi_t = C'(q^*) x_t - C(q'_t + x_t) + C(q'_t) .
  \]
  By definition of the target strategy $\s^*$ and
  Lemma~\ref{lem:profittrade}, $\pi_t \geq 0$ for all $t$.

  Consider a round $t$ in which $|q'_t - q^*| \geq \gamma$.  Suppose
  first that $q'_t \geq q^* + \gamma$, so a trader playing the target
  strategy would sell.  By definition of $\s^*$, $x_t = -\min\{q'_t -
  q^*, k\} \leq -\gamma$.  We can write
  \begin{align*}
     \pi_t &= C'(q^*) (x_t + \gamma) - C'(q^*) \gamma 
     - C(q'_t + x_t) + C(q'_t - \gamma) - C(q'_t - \gamma) + C(q'_t)
     \\ 
     &\geq - C'(q^*) \gamma - C(q'_t - \gamma) + C(q'_t)
     \\
     &  \geq - C'(q^*) \gamma - C(q^*) + C(q^* + \gamma) 
     \\
     & = D_C(q^* + \gamma, q^*) \geq \chi,
  \end{align*}
  where $\chi$ is defined as in the lemma statement.  The first
  inequality follows from an application of
  Lemma~\ref{lem:profittrade} with $q = q'_t - \gamma$ and $x = x_t +
  \gamma$.  The second follows from the convexity of $C$ and the
  assumption that $q'_t \geq q^* + \gamma$.

  If instead $q'_t \leq q^* - \gamma$ (so a trader playing the target
  strategy would buy), a similar argument can be made to show that
    $\pi_t \geq D_C(q^* - \gamma, q^*) \geq \chi$. 

  Putting this all together, we have
  \begin{align*}
    \E\left[L_T(x_1, \ldots, x_T, \eta_1, \ldots, \eta_T,
      \omega)\right] 
    &= \sum_{t=1}^T \E[\pi_t]
    \geq \sum_{t=1}^T \chi \Pr( |q'_t - q^* | \geq \gam)
  \end{align*}
  as desired.  The fact that $\chi > 0$ follows from the fact that it
  is the minimum of two Bregman divergences, each of which is strictly
  positive since $C$ is nonlinear (and thus strictly convex) in the
  region around $q^*$ and the arguments are not equal.
\end{proof}

\end{document}